\documentclass[conference]{IEEEtran}
\usepackage[explicit]{titlesec}
\usepackage{graphicx}
\usepackage{verbatim}
\usepackage{epstopdf}
\usepackage{setspace}
\usepackage{subfigure} 
\usepackage{algorithmic}
\usepackage{amsmath}
\usepackage{amsthm}
\usepackage{mathrsfs}
\usepackage{extarrows}
\usepackage{cite}
\newtheorem{theorem}{Theorem}

\usepackage{bbm}
\usepackage[bookmarks=false, colorlinks=true, citecolor=blue]{hyperref}
\usepackage{amssymb}

\newcommand{\RNum}[1]{\uppercase\expandafter{\romannumeral #1\relax}}
\newtheorem{lemma}{Lemma}

\newtheorem{definition}{Definition}
\newtheorem{corollary}{Corollary}
\usepackage{stfloats}

\usepackage{color, soul}
\usepackage{booktabs}
\usepackage{breqn}
\usepackage{tikz,xcolor}

\definecolor{lime}{HTML}{A6CE39}
\usepackage[left=0.68in,right=0.68in,top=0.73in,bottom=1.05in]{geometry}

\titlespacing{\section}{0pt}{1.2ex plus .0ex minus .0ex}{.3ex plus .0ex}
\titlespacing{\subsection}{0pt}{1.2ex plus .0ex minus .0ex}{.3ex plus .0ex}

\makeatletter
\DeclareRobustCommand{\orcidicon}{%
	\begin{tikzpicture}
		\draw[lime, fill=lime] (0,0) 
		circle [radius=0.16] 
		node[white] {{\fontfamily{qag}\selectfont \tiny ID}};    \draw[white, fill=white] (-0.0625,0.095) 
		circle [radius=0.007];    \end{tikzpicture}
	\hspace{-2mm}}
\foreach \x in {A, ..., Z}{%
	\expandafter\xdef\csname orcid\x\endcsname{\noexpand\href{https://orcid.org/\csname orcidauthor\x\endcsname}{\noexpand\orcidicon}}
}

\newcommand*\bigcdot{\mathpalette\bigcdot@{.5}}
\newcommand*\bigcdot@[2]{\mathbin{\vcenter{\hbox{\scalebox{#2}{$\m@th#1\bullet$}}}}}
\makeatother

\usepackage[linesnumbered,ruled,vlined]{algorithm2e}
	
	\IEEEoverridecommandlockouts

\begin{document}
		\title{Sampling to Achieve the Goal: An Age-aware Remote Markov Decision Process \vspace{-0.5em}}
		\author{Aimin Li\IEEEauthorrefmark{1}\IEEEauthorrefmark{3},
			Shaohua Wu\IEEEauthorrefmark{1}\IEEEauthorrefmark{2},
			Gary C.F. Lee\IEEEauthorrefmark{3}, 
			Xiaomeng Chen\IEEEauthorrefmark{1}, 
			and Sumei Sun\IEEEauthorrefmark{3},
			\emph{Fellow, IEEE}\\
			\IEEEauthorrefmark{1}\textit{Harbin Institute of Technology (Shenzhen), China} \\
			\IEEEauthorrefmark{3}\textit{Institute for Infocomm Research, Agency for Science, Technology and Research, Singapore}\\
			\IEEEauthorrefmark{2}\textit{The Department of Boradband Communication, Peng Cheng Laboratory, Shenzhen, China}\\
			
			\textit{E-mail: \{liaimin,23s052026\}@stu.hit.edu.cn; hitwush@hit.edu.cn; \{Gary\_Lee,sunsm\}@i2r.a-star.edu.sg}
			\vspace{-2em}
			
			\thanks{
				This work has been supported in part by the National Key Research and Development Program of China under Grant no. 2020YFB1806403, and in part by the Guangdong Basic and Applied Basic Research Foundation under Grant no. 2022B1515120002, and in part by the Major Key Project of PCL under Grant no. PCL2024A01.
			}
		}
		
		\maketitle
		\allowdisplaybreaks
		
		\begin{abstract}
			Age of Information (AoI) has been recognized as an important metric to measure the freshness of information. Central to this consensus is that minimizing AoI can enhance the freshness of information, thereby facilitating the accuracy of subsequent decision-making processes. However, to date the direct causal relationship that links AoI to the utility of the decision-making process is unexplored. To fill this gap, this paper proposes a sampling-control co-design problem, referred to as an \textit{age-aware remote Markov Decision Process} (MDP) problem, to explore this unexplored relationship. Our framework revisits the sampling problem in \cite{DBLP:journals/tit/SunUYKS17} with a refined focus: moving from AoI penalty minimization to directly optimizing goal-oriented remote decision-making process under random delay. We derive that the \textit{age-aware remote MDP problem} can be reduced to a standard MDP problem without delays, and reveal that treating AoI solely as a metric for optimization is not optimal in achieving remote decision making. Instead, AoI can serve as important \textit{side information} to facilitate remote decision making. 
			
			The code is available at \href{https://github.com/AiminLi-Hi/Age-Aware-Remote-MDP}{\textcolor{black}{\textit{https://github.com/AiminLi-Hi/Age-Aware-Remote-MDP}}}.
		\end{abstract}
		
		\IEEEpeerreviewmaketitle
		
		\section{Introduction}\label{sectionI}
		
		Markov Decision Process (MDP) has been a general framework for treating the sequential stochastic control problem \cite{howard1960dynamic,bellman1966dynamic}, and has been applied as an efficient theoretical framework for healthcare management, transportation scheduling, industrial production and automation, response and rescue systems, financial modeling, and \textit{etc}. \cite{boucherie2017markov}. Typically, a standard MDP framework assumes immediate access to the current state information, and the decision maker chooses actions based on the available \textit{delay-free} state of the system to achieve a specific goal. This idealization, however, may not hold in many practical scenarios. For instance, in a remote healthcare management system, the monitored patient's condition might be delayed for subsequent healthcare operations. In the Industrial Internet of Things, the transmission of critical safety data to the decision center might be subject to various network delays. These highlight the need for extending the standard MDP to the MDP with observation delays \cite{DBLP:conf/sigmetrics/AltmanN92,DBLP:journals/tac/KatsikopoulosE03}.
		

		There are two types of MDP that consider the observation delay, termed deterministic delayed MDP (DDMDP) \cite{DBLP:conf/sigmetrics/AltmanN92} and stochastic delayed MDP (SDMDP) \cite{DBLP:journals/tac/KatsikopoulosE03}. The DDMDP introduces a constant observation delay $d$ to the standard MDP framework. At any given time $t$, the decision-maker accesses the time-varying data as $O(t)=X_{t-d}$. The main result of the DDMDP problem is its reducibility to a standard MDP without delays through \textit{state augmentation}, as detailed by Altman and Nain \cite{DBLP:conf/sigmetrics/AltmanN92}. The SDMDP extends DDMDP by treating the observation delay not as a static constant but as a random variable $D$ following a given distribution $\Pr(D=d)$, with $O(t)=X_{t-D}$. In 2003, V. Katsikopoulos and E. Engelbrecht showed that an SDMDP is also reducible to a standard MDP problem without delay \cite{DBLP:journals/tac/KatsikopoulosE03}. Thus, it becomes clear to solve a SDMDP problem by solving its equivalent standard MDP.
		
		\begin{table}[t]\label{1}
			\centering
			\caption{Comparisons of Time-Lag MDPs}
			\begin{tabular}{ccc}
				\toprule
				\textbf{Type} & \textbf{Observation} & \textbf{Reference} \\
				\midrule
				MDP Without Delay & \( O(t) = X_t \) & \cite{bellman1966dynamic}\\
				DDMDP & \( O(t) = X_{t-d} \) & \cite{DBLP:conf/sigmetrics/AltmanN92} \\
				SDMDP & \( O(t) = X_{t-D} \) & \cite{DBLP:journals/tac/KatsikopoulosE03}\\
				Age-Aware Remote MDP & \( O(t) = X_{t-\Delta(t)} \) & {This Work}\\
				\bottomrule
			\end{tabular}
			\vspace{-4mm}
		\end{table}
		
		However, the above time-lag MDPs, where the observation delay follows a given distribution, potentially assumes that the state is sampled and transmitted to the decision maker at every time slot\footnote{In this case, each state $X_{i}, \forall i \in \{0,1,...n\}$ are all sampled and forwarded to the decision maker. The observation delay $D$ is i.i.d random variable and is independent of the sampling policy.}. This setup presumes that the system can transmit every state information without encountering any \textit{backlog}. In practice, constantly sampling and transmitting may result in infinitely accumulated packets in the queue, resulting in severe congestion. This motivates the need for queue control and adaptive sampling policy design in the network \cite{DBLP:conf/isit/Yates15,DBLP:journals/tit/SunUYKS17,DBLP:journals/tcom/ArafaBSP21,DBLP:journals/tit/TangCWYT23,panjiayu2023,BZJSSU}, where Age of Information (AoI) has emerged as an important indicator \cite{DBLP:conf/infocom/KaulYG12,DBLP:journals/ftnet/KostaPA17,DBLP:journals/jsac/YatesSBKMU21a,10105150}. {Currently, AoI has been applied in a wide range of applications such as queue control \cite{costa2016age,DBLP:journals/tcom/DoganA21,yates2018age,kam2015effect}, remote estimation \cite{sun2019samplingwiener,ornee2019sampling,DBLP:journals/tcom/ArafaBSP21,10.1145/3492866.3549732,tsai2021unifying,ornee2021sampling}, and communications \& network design\cite{li2022age,pan2022age,xie2020age,meng2022analysis,DBLP:journals/tcom/CaoZJWS21,DBLP:conf/globecom/Long0GLN22,DBLP:journals/tcom/FengWFCD24,DBLP:journals/tmc/PanCLL23,10539623}}. Suppose the $i$-th sample is generated at time $S_i$ and is delivered at the receiver at time $D_i$, AoI is defined as a \textit{sawtooth piecewise function}: \begin{equation}\label{eq1}\Delta(t)=t-S_i, D_{i}\le t< D_{i+1}, \forall i\in\mathbb{N},\end{equation}
		as shown in Fig. \ref{fig:figure1}. From this definition, the most recently available information at the receiver at time slot $t$ is $O(t)=X_{t-\Delta(t)}$. Different from the DDMDP and SDMDP where the time lag is a constant $d$ or an i.i.d random variable $D$, with $O(t)=X_{t-d}$ or $O(t)=X_{t-D}$, the time lag in the practical network with queue and finite transmission rate is a controlled random process $\Delta(t)$, with $O(t)=X_{t-\Delta(t)}$. See Table \ref{1} and Fig. \ref{fig:figure2} for the comparisons of different time-lag MDPs.
		
		\begin{figure}[!t]
			\centering
			\includegraphics[width=0.82\linewidth]{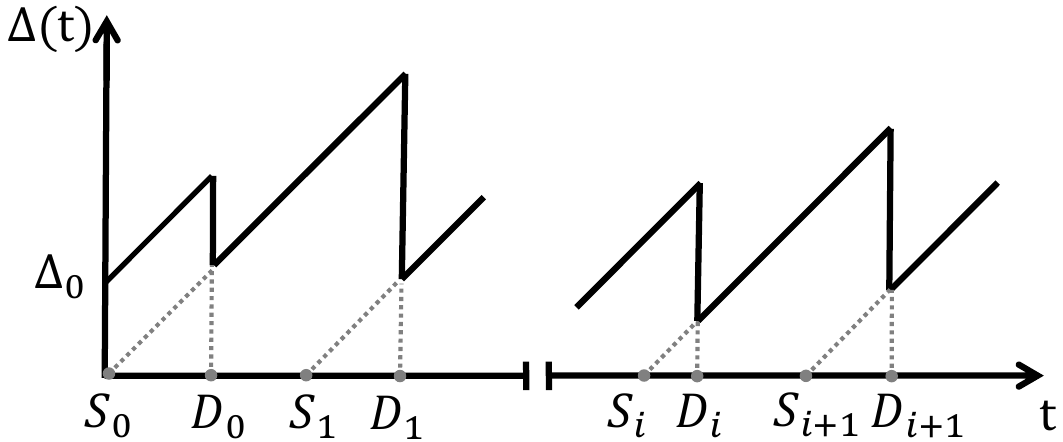}
			\caption{Evolution of the age $\Delta(t)$ over time.}
			\label{fig:figure1}
			\vspace{-3.5mm}
		\end{figure}
		
		Motivated by the above, this work enriches the time-lag MDP family by treating the observation delay not as a stationary stochastic variable $d$ or $D$, but as a dynamically controlled stochastic process, represented by \textit{age} $\Delta(t)$ defined in (\ref{eq1}). We refer to this problem as \textit{age-aware remote MDP} problem\footnote{In \cite{10273599}, the term \textit{remote MDP} was first proposed as a pathway to pragmatic or goal-oriented communications. Our paper follows this research and focuses on the effect from \textit{age}, hence the term \textit{age-aware remote MDP}.}, where AoI serves as no longer a typical indicator but important \textit{side information} for goal-oriented remote decision making. Our main result is that \textit{age-aware remote MDP}, like DDMDP \cite{DBLP:conf/sigmetrics/AltmanN92} and SDMDP \cite{DBLP:journals/tac/KatsikopoulosE03}, can also be reduced to a standard MDP problem with a constraint. We design efficient algorithms to solve this type of standard MDP with a constraint. \emph{To the best of our knowledge, this is the first work that introduces AoI into the time-lag MDP family, and the first work that explores AoI's new role as \textit{side information} to facilitate remote decision making under random delay.}

		\section{System Model and Problem Formulation}\label{section III}
		
		
		We consider a time-slotted \textit{age-aware remote MDP} problem illustrated in Fig. \ref{fig:figure2}(c). Let $X_t\in\mathcal{S}$ be the controlled source at time slot $t$. The evolution of the source is a Markov decision process, characterized by the transition probability $\Pr(X_{t+1}|X_t,a_t)$\footnote{For short-hand notations, we use the transition probability matrix $\mathbf{P}_a$ to encapsulate the dynamics of the source given an action $a_t=a$.}, where $a_t\in\mathcal{A}$ represents the action taken by the remote decision maker to control the source in the desired way. The sampler conducts the sampling action $a_t^S\in\{0,1\}$, with $a_t^S=1$ representing the sampling action and $a_t^S=0$ otherwise. Consider the random channel delay of the $i$-th packet as $Y_i\in\mathcal{Y}\subseteq\mathbb{N}^+$, which is independent of the source $X_t$ and is bounded $\max[{Y_i}]<\infty$. The sampling times $S_0,S_1,\cdots$ shown in Fig. \ref{fig:figure1} record the time stamp with $a_t^S=1$, 
		\begin{equation}\label{eq2}
			S_i=\max\{t\left|t\le D_i,a_t^S=1\right.\},
		\end{equation}
		where the initial state of the system is $S_0=0$ and $\Delta(0)=\Delta_0$.
		
		\begin{figure}[tbp]
			\centering
			\includegraphics[width=0.92\linewidth]{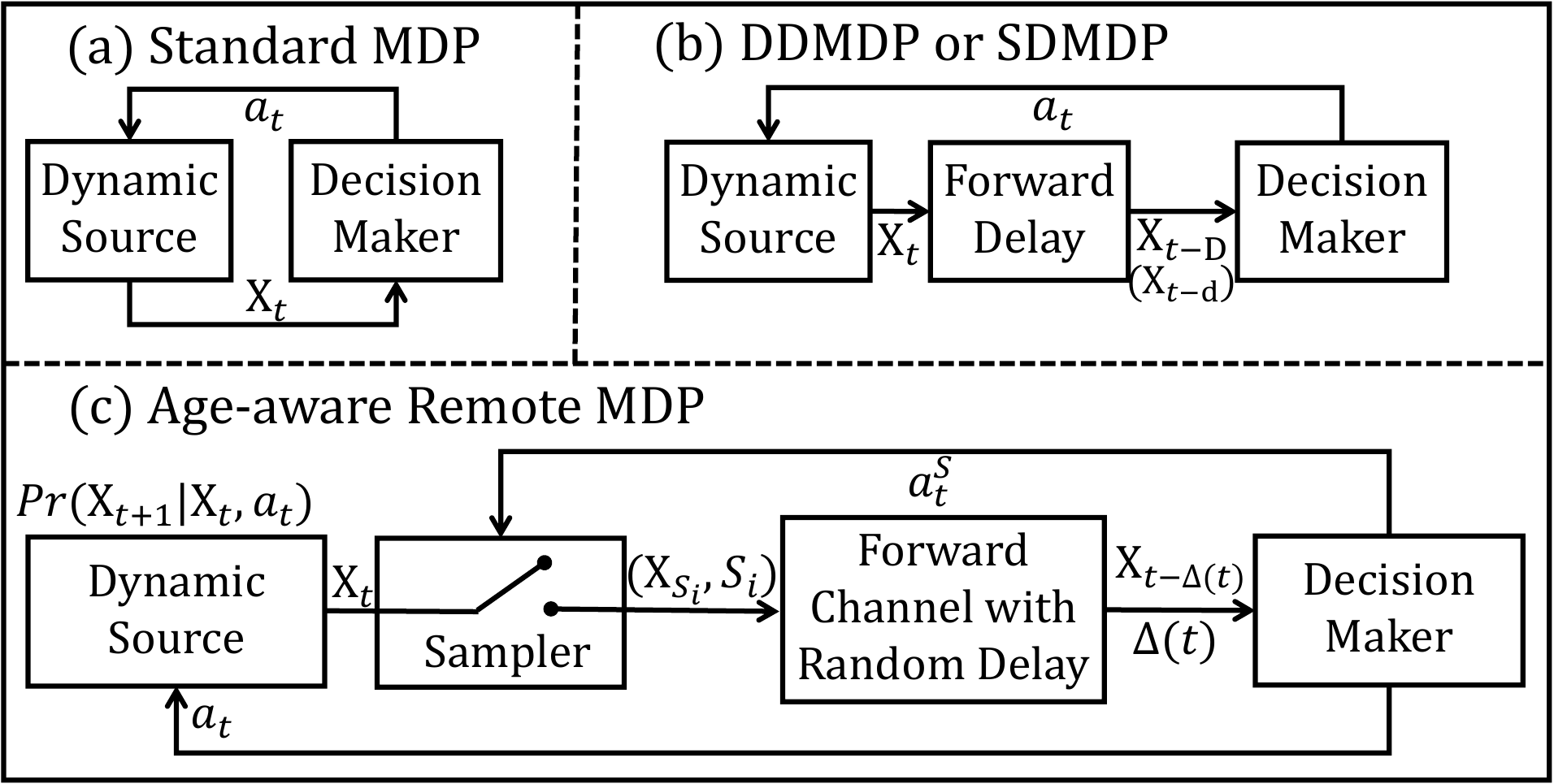}
			\caption{Comparisons among standard MDP, DDMDP, SDMDP, and \textit{age-aware remote MDP}.}
			\label{fig:figure2}
			\vspace{-5mm}
		\end{figure}
		
		At the sampling time $S_i,\forall i\in\mathbb{N}$, the state $X_{S_i}$ along with the corresponding time stamp $S_i$ are encapsulated into a packet $(X_{S_i},S_i)$, which is transmitted to a remote decision maker. Upon the reception of the packet $(X_{S_i},S_i)$ during $t\in[D_i,D_{i+1})$, the \textit{ observation history} at the decision maker is $\{(X_{S_j},S_j):{j\le i}\}$. By employing (\ref{eq1}), this sequence can be equivalently represented as $\{(X_{k-\Delta(k)},\Delta(k)):{k\le t}\}$.
		
		Denote $\mathcal{I}_t=\{(X_{k-\Delta(k)},\Delta(k),a_{k-1}^S,a_{k-1}):{k\le t}\}$ as the \textit{observation history} and \textit{action history} (or simply \textit{history}) available to the decision maker up to time $t$. The decision maker is tasked with determining both the sampling action and the controlled action $(a_t^S,a_t)\in\{0,1\}\times\mathcal{A}$ at each time slot $t$ by leveraging the \textit{history} $\mathcal{I}_t$. A decision policy of the decision maker is defined as a mapping from the \textit{history} to a distribution over the joint action space $\{0,1\}\times\mathcal{A}$, denoted by $\pi_t: \mathcal{I}_t\rightarrow\mathrm{Prob}( \mathcal{A}\times\{0,1\})$. Similar to \cite{DBLP:journals/tit/SunUYKS17}, we assume that the sampling policy satisfies two conditions: 
		\begin{itemize}
			\item $i$) No sample is taken when the channel is busy, \emph{i.e.}, $S_{i+1}\ge D_i$, or equivalently, \begin{equation}\label{eq4}
				S_{i+1}-Z_i=D_i \text{ with } Z_i\ge0, i\in\mathbb{N},
			\end{equation}
			with $Z_i$ representing the sampling waiting time. From this assumption, we can also obtain that $D_i=S_i+Y_i,\forall i\in\mathbb{N}$.
			\item $ii$) The inter-sample times $G_i=S_{i+1}-S_i$ is a \textit{regenerative process} \cite[Section 6.1]{haas2006stochastic}: There is a sequence $0\le g_1<g_2<\cdots$ of almost surely finite random integers such that the post-$g_j$ process $\{G_{{g_j}+i}\}_{i\in\mathbb{N}}$ has the same distribution as the post-$g_1$ process $\{G_{{g_1}+i}\}_{i\in\mathbb{N}}$ and is independent of the pre-$g_j$ process $\{G_{{g_1}+i}\}_{i\in\{0,1,\cdots k_j-1\}}$. Condition $ii$) implies that, almost surely \footnote{This assumption also implies that the waiting time $Z_i$ is bounded, belonging to a subset of nature numbers with $Z_i\in\mathcal{Z}\subseteq\mathbb{N}$.}
			\begin{equation}\label{eq3}
				\lim\limits_{i\rightarrow \infty}S_i=\infty, ~~~\lim\limits_{i\rightarrow \infty}D_i=\infty.
			\end{equation}
		\end{itemize}
		
		In addition, we assume that the controlled policy satisfies the following condition: the action $a_t$ is updated only upon the delivery of a sample $X_{S_i}$\footnote{Throughout the time interval $t \in [D_i, D_{i+1})$, the decision maker's observations are fixed and consist only of ${(X_{S_j},S_j) : j \le i}$. Thus, it is assumed that the chosen actions remain constant in this period. The possibility of varying these actions within such intervals will be our future work.}, \emph{i.e.}, \begin{equation}a_t= A_i, D_i\le t \le D_{i+1}, i\in\mathbb{N},\end{equation} where $A_i\in\mathcal{A}$ is the updated controlled action upon the delivery of packet $(X_{S_i},S_i)$. We consider bounded cost function $\mathcal{C}(X_t,a_t)<\infty$, which represents the immediate cost incurred when action $a_t$ is taken in state $X_t$. Under the above assumptions, the objective of the system is to design the optimal policies at each time slot $\pi_0, \pi_1,\pi_2\cdots$ to minimize the \textit{long-term average cost} \footnote{In Theorem \ref{the1}, we will discuss the unichain property of the MDP, therefore the problem is independent of the initial state $X_0$.}:
		\begin{equation}\label{eq5}
			\mathcal{P}1: \inf _{\pi_{0: \infty}} \limsup _{T \rightarrow \infty} \frac{1}{T} \mathbb{E}\left[\sum_{t=1}^T \mathcal{C}(X_t, a_t)\right].
		\end{equation} 
		This problem aims at determining the distribution of joint sampling and controlled actions $(a_t^S,a_t)$ based on the history $\mathcal{I}_t$, such that the \textit{long-term average cost} is minimized.
		
		
		\section{Optimal Sampling and Remote Decision Making Policy Under Random Delay}\label{sectionIV}
		
		\subsection{Sufficient Statistics of History}
		The policy $\pi_t$ is a mapping from $\mathcal{I}_t$ to the action provability space. One challenge to solving $\pi_t$ is that the \textit{history} space $\mathcal{I}_t$ explodes exponentially as $t$ increases. This motivates us to compress and abstract information that is necessary for the optimal decision process. Specifically, we will analyze the \textit{sufficient statistics} in this subsection.
		\begin{definition}\label{definition1}
			A sufficient statistics of $\mathcal{I}_t$ is a function $S_t(\mathcal{I}_t)$, such that $\min_{a_{t:T}}\mathbb{E}\left[\sum_{k=t}^T \mathcal{C}(X_k, a_k)|\mathcal{I}_t\right]=\min_{a_{t:T}}\mathbb{E}\left[\sum_{k=t}^T \mathcal{C}(X_k, a_k)|S_t(\mathcal{I}_t)\right]$ holds for any $T>t$. 
		\end{definition} 
		The above definition implies that the decision making based on the \textit{sufficient statistics} $H_t(\mathcal{I}_t)$ can achieve an equivalent performance as that dependent on $\mathcal{I}_t$. The following lemma introduces an important \textit{sufficient statistics} of $\mathcal{I}_t$.
		\begin{lemma}\label{l1}
			In our problem, during the interval $t\in [D_{i},D_{i+1})$, $\mathcal{G}_i=(X_{S_i},Y_i,A_{i-1})\in\mathcal{S}\times\mathcal{Y}\times\mathcal{A}$ is a sufficient statistics of $\mathcal{I}_t$. Besides, determining the optimal sampling actions $a_t^S$ under condition (\ref{eq4}) is equivalent to determining the optimal sampling time $S_{i+1}$, or the optimal waiting time $Z_i$.
		\end{lemma}
		\noindent \textit{Proof.} 
		See Appendix \ref{appendixa} for the proof. $\blacksquare$\vspace{2mm}
		
		With Lemma \ref{l1}, determining the optimal policy $\pi_t:\mathcal{I}_t\rightarrow\mathrm{Prob}(\mathcal{A}\times\{0,1\})$ for $\mathcal{P}1$ becomes equivalent to solving for an alternative policy $\phi_t: \mathcal{S}\times\mathcal{Y}\times\mathcal{A}\rightarrow\mathrm{Prob}(\mathbb{N}\times\mathcal{A})$, which is a mapping from $\mathcal{G}_i$ to a distribution of $(Z_i,A_i)$. The problem $\mathcal{P}1$ is thus rewritten as
		\begin{equation}
			\mathcal{P}2: \inf _{\phi_{0: \infty}} \limsup _{T \rightarrow \infty} \frac{1}{T} \mathbb{E}\left[\sum_{t=1}^T \mathcal{C}(X_t, a_t)\right].
		\end{equation}
		
		\subsection{Simplification of Problem $\mathcal{P}2$}
		
		As $G_i$ is a \textit{regenerative process} and $\lim\limits_{i\rightarrow \infty}D_i=\infty$, we can rewrite the objective function in problem $\mathcal{P}2$ as 
		\vspace{-1mm}
		\begin{align}
			\limsup _{T \rightarrow \infty} \frac{1}{T} \mathbb{E}\left[\sum_{t=1}^T \mathcal{C}(X_t, a_t)\right]\nonumber
		\end{align}
		\vspace*{0.03mm}
		\begin{align}
			&= \lim _{\it{D}_{\mathrm{n}} \rightarrow \infty} \frac{\mathbb{E}\left[\sum_{t=1}^{\it{D}_{\mathrm{n}}} \mathcal{C}(X_t, a(t))\right]}{\mathbb{E}\left[\it{D}_{\mathrm{n}}\right]}\nonumber\\
			&= \lim _{\mathrm{n} \rightarrow \infty} \frac{\sum_{i=0}^{n-1} \mathbb{E}\left[\sum_{t=\it{D}_{i}}^{\it{D}_{i+1}-1} \mathcal{C}(X_t, A_i)\right]}{\sum_{i=0}^{n-1} \mathbb{E}\left[\it{D}_{i+1}-\it{D}_{i}\right]}\nonumber\\
			&=\lim _{\mathrm{n} \rightarrow \infty} \frac{\sum_{i=0}^{n-1} \mathbb{E}\left[\sum_{t=\it{D}_{i}}^{\it{D}_{i+1}-1} \mathcal{C}(X_t, A_i)\right]}{\sum_{i=0}^{n-1} \mathbb{E}\left[Y_{i+1}+Z_i\right]}.
		\end{align}
		Then Problem $\mathcal{P}2$ can be rewritten as
		\begin{equation}
			\mathcal{P}3:h^* \triangleq \inf _{\phi_{0: \infty}}\lim _{\mathrm{n} \rightarrow \infty} \frac{\sum_{i=0}^{n-1} \mathbb{E}\left[\sum_{t=\it{D}_{i}}^{\it{D}_{i+1}-1} \mathcal{C}(X_t, A_i)\right]}{\sum_{i=0}^{n-1} \mathbb{E}\left[Y_{i+1}+Z_i\right]}.
		\end{equation}
		
		To solve Problem $\mathcal{P}3$, we consider the following problem with parameter $\lambda\ge0$: 
		\begin{equation}\resizebox{1\hsize}{!}{$
				\begin{aligned}
					&\mathcal{P}4: U(\lambda)\triangleq\\
					&\inf _{\phi_{0: \infty}} \lim _{\mathrm{n} \rightarrow \infty}\frac{1}{n} {\sum_{i=0}^{n-1}\left\{ \mathbb{E}\left[\sum_{t=\it{D}_{i}}^{\it{D}_{i+1}-1} \mathcal{C}(X_t, A_i)\right]-\lambda\mathbb{E} \left[Z_i+Y_{i+1}\right]\right\}},
				\end{aligned}$}
		\end{equation}
		
		By similarly applying Dinkelbach’s method as in \cite{dinkelbach1967nonlinear} and \cite[Lemma 2]{sun2019samplingwiener}, we obtain the following lemma:
		\begin{lemma}\label{l3}
			The following assertions hold:\\
			(i). $h^* \gtreqless \lambda \text { if and only if } U(\lambda) \gtreqless 0 \text {. }$\\
			(ii). When $U(\lambda)=0$, the solutions to Problem $\mathcal{P}4$ coincide with those of Problem $\mathcal{P}3$.\\
			(iii). $U(\lambda)=0$ has a unique root, and the root is $h^*$.
		\end{lemma}
		\noindent \textit{Proof.} 
		See Appendix \ref{appendixb} for the proof. $\blacksquare$\vspace{2mm}
		
		Following Lemma \ref{l3}, solving Problem $\mathcal{P}3$ can be simplified to solving Problem $\mathcal{P}4$ under $U(\lambda)=0$. Our remaining goal is to solve $\mathcal{P}4$ and search $h^*$ such that $U(h^*)=0$.
		\subsection{Reformulate Problem $\mathcal{P}4$ as a Standard MDP}\label{MDP}
		In this subsection, we formulate Problem $\mathcal{P}4$ as a standard infinite horizon MDP problem. We introduce the state space, action space, transition probability, and cost function of the MDP problem in this subsection specifically. This MDP problem with parameter $\lambda$ is denoted as $\mathscr{P}_{\mathrm{MDP}}(\lambda)$:
		\begin{itemize}
			\item \textbf{State Space}: the state of the equivalent MDP is the sufficient statistics $\mathcal{G}_i=(X_{S_i},Y_i,A_{i-1})\in\mathcal{S}\times\mathcal{Y}\times\mathcal{A}$.
			\item \textbf{Action Space}: the actions space of the MDP composed by the tuple $(Z_i,A_i)\in\mathcal{\mathcal{Z}\times\mathcal{A}}$, where $Z_i$ is the sampling waiting time and $A_i$ is the controlled actions. 
			\item \textbf{Transition Probability}: The transition probability is defined by $\Pr(\mathcal{G}_{i+1}|\mathcal{G}_i,Z_i,A_i)$. We have the transition probability as ({see Appendix \ref{appendixc} for the proof}):
			\begin{equation}\label{tran}
				\begin{aligned}
					&\Pr(\mathcal{G}_i=(s',\delta',a')|\mathcal{G}_i=(s,\delta,a),Z_i,A_i)\\
					&=\Pr(Y_{i+1}=\delta')\cdot[\mathbf{P}_a^{\delta}\cdot\mathbf{P}_{a_i}^{Z_i}]_{s\times s'}\cdot\mathbbm{1}\{a'=A_i\},
				\end{aligned}
			\end{equation}
			
			\item \textbf{Cost Function}: the cost function is typically a real-valued function over the state space and the action space. We denote the cost function as $g(\mathcal{G}_i,Z_i,A_i)$, and next show that we can tailor the cost function to establish an identical standard MDP of Problem $\mathcal{P}4$. 
			\begin{lemma}\label{l4}
				If the cost function is defined by 
				\begin{equation}\label{cfun}
					\begin{aligned}
						g(\mathcal{G}_i,Z_i,A_i;\lambda)\triangleq q(\mathcal{G}_i,Z_i,A_i)-\lambda f(Z_i),
					\end{aligned}
				\end{equation}
				\begin{align}
					\text{where          }~~~~f(Z_i)=Z_i+\mathbb{E}\left[Y_{i+1}\right],&
				\end{align}
				\begin{equation}\label{15}\resizebox{0.92\hsize}{!}{$
						\begin{aligned}
							&q(\mathcal{G}_i,Z_i,A_i)=q(X_{S_i},Y_i,A_{i-1},Z_i,A_i)\\
							&=\mathbb{E}\left[\sum_{s'\in\mathcal{S}}\left[\sum_{t=0}^{Z_i+\mathrm{Y}_{i+1}-1} \mathbf{P}_{A_{i-1}}^{Y_i}\cdot\mathbf{P}_{A_{i}}^{t}\right]_{X_{S_i}\times s'}\cdot\mathcal{C}(s',A_i)\right],
						\end{aligned}$}
				\end{equation}
				then Problem $\mathscr{P}_{\mathrm{MDP}}(\lambda)$ is identical to Problem $\mathcal{P}4$.
			\end{lemma}
		\end{itemize}
		\noindent \textit{Proof.} 
		See Appendix \ref{appendixd} for the proof. $\blacksquare$\vspace{2mm}
		
		Our remaining focus is to solve $\mathscr{P}_{\mathrm{MDP}}(\lambda)$ and seek a value $h^*$ such that $U(h^*)=0$.

		\subsection{Existence of Optimal Stationary Deterministic Policy}
		We examine the sufficient conditions required for the existence of a \textit{stationary deterministic} policy within $\mathscr{P}_{\mathrm{MDP}}(\lambda)$. Our main result is described in the following theorem:
		

		\begin{theorem}\label{the1}
			If an MDP characterized by finite state space $\mathcal{S}$, finite action space $\mathcal{A}$, and transition probability $\mathbf{P}_a$ is a unichain, then: ($i$) the transformed Age-aware remote MDP $\mathscr{P}_{\mathrm{MDP}}(\lambda)$ is also a unichain; ($ii$) an optimal stationary deterministic policy exists for $\mathscr{P}_{\mathrm{MDP}}(\lambda)$. 
		\end{theorem}
		\noindent \textit{Proof.} 
		See Appendix \ref{appendixe} for the proof. $\blacksquare$ \vspace{2mm}
		
		\subsection{Numerical Solutions}
		We propose two algorithms to solve the infinite-horizon MDP and seek the parameter $h^*$ such that $U(h^*)=0$. 
		
		$\bullet$ \textit{Bisec-MRVI}: The first one is a \textit{two-layer} algorithm. The outer layer is based on the \textit{bisection-search} method: the search interval $(\lambda_{\downarrow}^{(k)},\lambda_{\uparrow}^{(k)})$ is iteratively narrowed down by half until the interval can closely approximate the value $h^*$ such that $U(h^*)=0$. The internal layer utilizes a \textit{modified Relative Value Iteration} (MRVI) \cite[Eq. 4.72]{bertsekas2012dynamic2} to compute the value $U(\lambda)$ by resolving the MDP $\mathscr{P}_{\mathrm{MDP}}(\lambda)$ \footnote{Different from a standard RVI algorithm, the MRVI algorithm ensures that the algorithm converges under a weaker condition that the Markov chain under a given stationary policy is \textit{periodic} \cite[Proposition 4.3.4]{bertsekas2012dynamic2}.}. A similar \textit{two-layer bisection-based} method has been introduced in \cite{DBLP:journals/tit/SunUYKS17} and \cite{sun2019samplingwiener} to achieve Age-optimal or Mean Square Error (MSE)-optimal sampling. We note the complexity of \textit{Bisec-RVI} algorithm is dependent on the initialization of the search interval $(\lambda_{\downarrow},\lambda_{\uparrow})$, and thus establish an upper and lower bound of $h^*$ here:
		\begin{lemma}\label{l5}
			The lower bound of $h^*$ is given by \begin{equation}h^*\ge\min_{s,a}\mathcal{C}(s,a).\end{equation}
			The upper bound of $h^*$ is given by
			\begin{equation}
				h^*\le\min_a\sum_{s\in\mathcal{S}}{\pi}_a(s)\cdot \mathcal{C}(s,a),
			\end{equation}
			where ${\pi}_a(s)$ represents the stationary distribution of state $s$, corresponding to the transition probability matrix $\mathbf{P}_a$.
		\end{lemma}
		\noindent \textit{Proof.} 
		See Appendix \ref{appendixf} for the proof. $\blacksquare$ \vspace{2mm}
		
		The \textit{Bisec-MRVI} algorithm is in Algorithm \ref{Algorithm 1}, where the details of the \textit{MRVI} is demonstrated in \cite[Eq. 4.72]{bertsekas2012dynamic2}.
		\vspace{-2mm}
		\begin{algorithm}
			\caption{\textit{Bisec-MRVI} algorithm}
			\label{Algorithm 1}
			\LinesNumbered
			\KwIn{Tolerence $\epsilon>0$, MDP $\mathscr{P}_{\mathrm{MDP}}(\lambda)$}
			Initialization: $\lambda_{\uparrow}=\min_a\sum_{s\in\mathcal{S}}{\pi}_a(s)\cdot \mathcal{C}(s,a)$, $\lambda_{\downarrow}=\min_{s,a}\mathcal{C}(s,a)$ \;
			\While {$\lambda_{\uparrow}-\lambda_{\downarrow}\ge \epsilon$}
			{
				$\lambda=(\lambda_{\uparrow}+\lambda_{\downarrow})/2$\;
				Run \textit{MRVI} to solve $\mathscr{P}_{\mathrm{MDP}}(\lambda)$ and calculate $U(\lambda)$\;
				\If{$U(\lambda)>0$}
				{
					$\lambda_{\downarrow}=\lambda$\;
				}
				\Else
				{$\lambda_{\uparrow}=\lambda$\;}			
			}
			\KwOut{$h^*=\lambda$}
		\end{algorithm}
		\vspace{-3mm}
		
		$\bullet$ \textit{Fixed-Point-Based Iteration (FPBI)}: The \textit{two-layer} \textit{Bisec-MRVI} algorithm requires repeatedly executing the MRVI algorithm in the inner layer, which is \textit{computation-intensive}. This motivates us to propose a \textit{one-layer} algorithm, termed FPBI algorithm, to avoid the computation overhead. Unlike \textit{Bisec-MRVI}, The \textit{one-layer} FPBI algorithm treats $U(h^*)=0$ as a constraint within the Markov Decision Process $\mathscr{P}_{\mathrm{MDP}}(h^*)$. Our main result is the following equivalent equations.	
		\begin{theorem}\label{the2}
			Solving the root finding problem of $U(\lambda)$ is equivalent to solving the following nonlinear equations:
			\begin{equation}\label{eqfunction}
				\left\{
				\begin{aligned}
					&W^*(\gamma)=\min _{A_i, Z_i}\big\{g(\gamma,A_i,Z_i;h^*)+\mathbb{E}[W^*\left(\gamma'\right)|\gamma,Z_i,A_i]\big\}\\
					&\text{for } \gamma\in\mathcal{S}\times\mathcal{Y}\times\mathcal{A},\\
					&h^*=\min _{A_i, Z_i}\left\{\frac{q\left(\gamma^{\mathrm{ref}}, A_i, Z_i\right)+
						\mathbb{E}[W^*\left(\gamma'\right)|\gamma^{\mathrm{ref}}, A_i, Z_i]}{f\left(Z_i\right)}\right\},\\
				\end{aligned}\right.
			\end{equation}
			where $\gamma^{\mathrm{ref}}\in\mathcal{S}\times\mathcal{Y}\times\mathcal{A}$ can be arbitrarily chosen.
		\end{theorem}
		\noindent \textit{Proof.} 
		See Appendix \ref{appendixg} for the proof. $\blacksquare$\vspace{2mm}

		In (\ref{eqfunction}), there are $|\mathcal{S}|\times|\mathcal{Y}|\times|\mathcal{A}|+1$ variables and an equal number of non-linear equations. In what follows we reveal (\ref{eqfunction}) forms a set of \textit{fixed-point equations}, which can be effectively solved through \textit{fixed-point iterations} \cite{hildebrand1987introduction}. 
		
		Let $\mathbf{W}^*$ denote the vector consisting of $W^*(\gamma)$ for all $\gamma\in\mathcal{S}\times\mathcal{Y}\times\mathcal{A}$, (\ref{eqfunction}) can be succinctly represented as:
		\begin{equation}\label{eqfunctiontransformation}
			\left\{
			\begin{aligned}
				&\mathbf{W}^*=T(\mathbf{W}^*,h^*)\\
				&h^*=H(\mathbf{W}^*).
			\end{aligned}\right.
		\end{equation}
		Then, substituting the second equation $h^*=H(\mathbf{W}^*)$ into the first equation of (\ref{eqfunctiontransformation}) yields $\mathbf{W}^*=T(\mathbf{W}^*,H(\mathbf{W}^*))$. Define $Q(\mathbf{W})\triangleq T(\mathbf{W},H(\mathbf{W}))$ for simplification. Consequently, as implied by (\ref{eqfunctiontransformation}), we have the \textit{fixed-point equation}: \begin{equation}
			\mathbf{W}^*=Q(\mathbf{W}^*),
		\end{equation}
		which can be solved by \textit{fixed-point iteration} in Algorithm \ref{Algorithm 2}.
		\vspace{-1mm}
		\begin{figure}[tbp]
			\centering
			\includegraphics[width=1\linewidth]{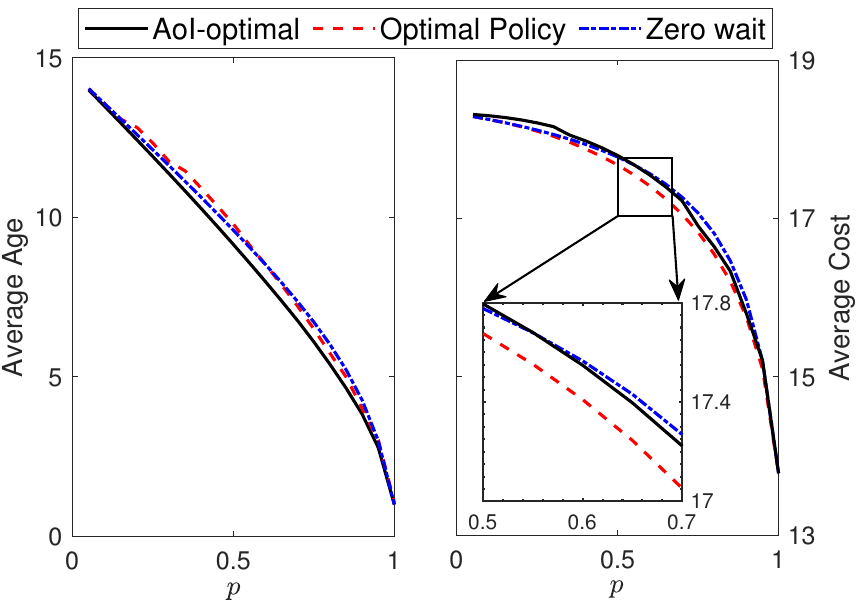}
			\caption{Average age and average cost vs. $p$ with \textit{i.i.d} random delay $Y_i$, where $\Pr(Y_i=1)=p$ and $\Pr(Y_i=10)=1-p$.}
			\label{fig:combinetwofigure}
			\vspace{-2mm}
		\end{figure}
		\begin{figure}[t]
			\centering
			\includegraphics[width=1\linewidth]{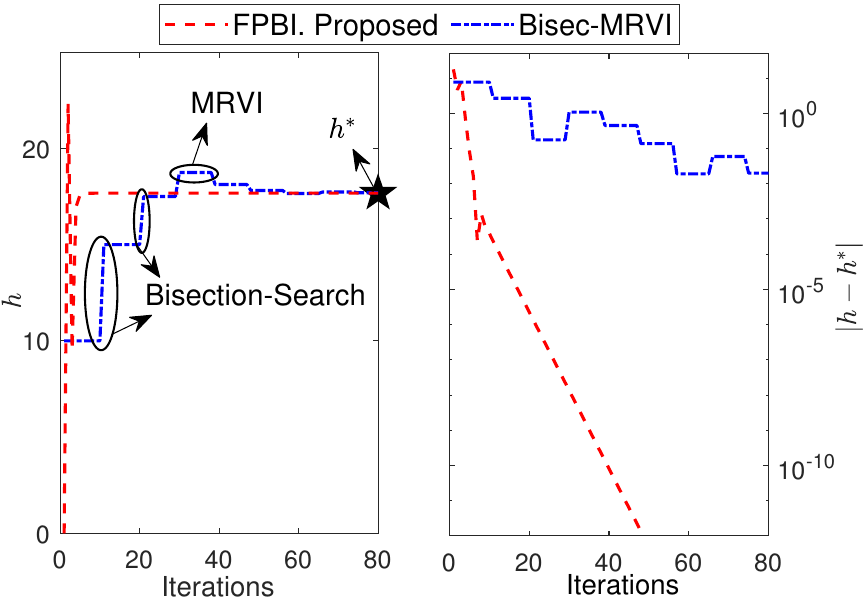}
			\caption{\textit{Bisec-MRVI} vs. \textit{FPBI}, where $p=0.5$.}
			\label{fig:converge}
			\vspace{-4mm}
		\end{figure}
		\vspace{-4mm}
		\begin{algorithm}
			\caption{FPBI for Solving (\ref{eqfunction})}
			\label{Algorithm 2}
			\LinesNumbered
			\KwIn{MDP $\mathscr{P}_{\mathrm{MDP}}(\lambda)$, Tolerence $\epsilon$;}
			Initialization: $\mathbf{W}^*_{(0)}=\mathbf{0}$, $\mathbf{W}^*_{(-1)}=\infty$, $k=0$ \;
			Choose $\gamma^{ref}$ arbitrarily\;
			\While {$||\tilde{V}_{\pi_A}^k(\mathbf{w})-\tilde{V}_{\pi_A}^{k-1}(\mathbf{w})||\ge \epsilon$}
			{
				$k=k+1$\;
				$h^*_{(k)}=H(\mathbf{W}^*_{(k-1)})$\;
				$\mathbf{W}^*_{(k)}=T(\mathbf{W}^*_{(k-1)},h^*_{(k)})$\;\tcp{See (\ref{eqfunction}) and (\ref{eqfunctiontransformation}) for $T(\cdot)$ and $H(\cdot)$.}	
			}
			\KwOut{$h^*=h^*_{(k)}$}
		\end{algorithm}
		\vspace{-1.5em}
		\section{Simulation Results}\label{sectionV}

		In this paper, the policy that minimizes the \textit{long-term average cost} in {Problem} $\mathcal{P}1$ is referred to as ``optimal policy''. We compare it with the following two benchmarks:
		\begin{itemize}
			\item \textit{Zero wait}: An update is transmitted once the previous update is delivered, \emph{i.e.}, $Z_i=0$ for $\forall i$. This policy achieves the minimum delay and maximum throughput.
			The controlled action $a_i$ is \textit{age-aware} and is obtained by substituting $Z_i=0$ into the RHS of (\ref{eqfunction}) and similarly implementing the \textit{fixed-point iteration}, implying \textit{age-aware optimal control} under \textit{zero-wait} sampling.
			\item \textit{AoI-optimal}: The \textit{AoI-optimal} policy determines $Z_i$ by \cite[Theorem 4]{DBLP:journals/tit/SunUYKS17}, which is a threshold-based policy $Z_i=\max(0,\beta-Y_i)$, where $\beta$ is numerically solved by \cite[Algorithm 2]{DBLP:journals/tit/SunUYKS17}. The controlled action $a_i$ is \textit{age-aware} and is obtained by fixing the \textit{AoI-optimal} sampling in (\ref{eqfunction}) and similarly implementing the \textit{fixed-point iteration}, implying \textit{age-aware optimal control} under \textit{AoI-optimal} sampling.
		\end{itemize}
		
		As a case study, we consider the parameters detailed in Appendix \ref{appendixh} for the simulation setup. 
		
		Fig. \ref{fig:combinetwofigure} compares the benchmarks with ``optimal policy'' in terms of average age and cost. The left panel demonstrates that the \textit{AoI-optimal} policy consistently achieves the lowest \textit{age}. However, the right panel reveals a \textit{counterintuitive} result: the \textit{AoI-optimal} policy does not necessarily lead to the best decision-making performance. This suggests that the \textit{value-of-information} transcends mere age freshness; it is also shaped by the specific goal of the receiver and the semantic content of the information \cite{wang2023review,9551200,10560514,9475174,9919752,DBLP:journals/corr/abs-2311-11143,10579545,10562359,10409276}. 
		
		Fig. \ref{fig:converge} compares \textit{FPBI} with \textit{Bisec-MRVI}, where the latter is the benchmark. This benchmark is inspired by the \textit{two-layer} solution frameworks by \cite{DBLP:journals/tit/SunUYKS17,sun2019samplingwiener} and \cite{sun2019sampling}. The left panel records the trajectories of $h$ across algorithms iterations, illustrating that both \textit{Bisec-MRVI} and \textit{FPBI} converge to the optimal $h^*$. For \textit{Bisec-MRVI}, updates to $h$ are deponent on the \textit{outer-layer bisection search}, which occur only upon the convergence of the \textit{inter-layer} MRVI. In contrast, our developed \textit{one-layer} \textit{FPBI} eliminates the need for the \textit{outer-layer} \textit{bisection search} by fixing $U(h^*)=0$ and directly reestablishing \textit{fixed point equations} for $\mathscr{P}_{\mathrm{MDP}}(h^*)$ under this condition. The right panel records the distances to $h^*$ across algorithms iterations, which demonstrated that \textit{FPBI} algorithm achieves faster convergence compared\textit{ Bisec-MRVI}.

		\section{Conclusion}\label{sectionVI}
		In this paper, we have proposed a new remote MDP problem in the time-lag MDP framework, termed \textit{age-aware remote MDP}. Specifically, AoI, typically an optimization indicator to ensure information \textit{freshness}, has been introduced into the remote MDP problem as a specific category of controlled random processing delay and as important \textit{side information} to enhance remote decision-making. The main result of this work is that the age-aware remote MDP can be reformulated into a standard, \textit{delay-free} MDP. Under such an equivalent problem, we have established sufficient conditions for the existence of an optimal \textit{stationary deterministic} policy. Additionally, we have developed a low-complexity \textit{one-layer} algorithm to effectively solve this remote MDP problem. We revealed that the age-optimal policy, which ensures the freshest information, does not necessarily achieve the best remote decision making. In contrast, we can design goal-oriented sampling policy that directly optimize remote decision making.
	
	\appendices
	\section{Proof of Lemma \ref{l1}}\label{appendixa}
	\subsection{Proof of Sufficient Statistics}
	We begin with rewriting the conditional expectation:
	\begin{equation}\label{21}
	\begin{aligned}
	&\max_{a_{t:T}}\mathbb{E}\left[\sum_{k=t}^T \mathcal{C}(X_k, a_k)|\mathcal{I}_t\right]\\
	&=\max_{a_{t:T}}\sum_{k=t}^T \mathcal{C}(s', a_{k})\cdot\Pr(X_k=s'\mid\mathcal{I}_t).
	\end{aligned}
	\end{equation}	
	Because $X_k$ is a Markov decision process, the conditional probability $\Pr(X_k=s'\mid\mathcal{I}_t)$, where $\mathcal{I}_t$ is defined as	$\mathcal{I}_t\triangleq\{(X_{k-\Delta(k)},\Delta(k),a_{k-1}^S,a_{k-1}):{k\le t}\}$, is only dependent on the most recently available information $X_{t-\Delta(t)}$, the age of the most recently available information $\Delta(t)$, and the actions taken from time slot $t-\Delta(t)$ to time slot $t$, $a_{t-\Delta(t):t}$. Thus, we can rewrite the conditional probability in (\ref{21}) as:  
	\begin{equation}\label{22}\begin{aligned}
	&\Pr(X_k=s'\mid\mathcal{I}_t)\\
	&=\Pr\left(X_k=s'\mid X_{t-\Delta(t)},\Delta(t),a_{t-\Delta(t):t-1}\right).
	\end{aligned}
	\end{equation}
	For $t$ within the interval $[D_{i},D_{i+1})$, equation (\ref{eq1}) implies that $t-\Delta(t)=S_i$. This allows us to deduce the right-hand side (RHS) of (\ref{22}) as
	\begin{equation}
	\begin{aligned}
	&\Pr\left(X_k=s'\mid X_{S_i},t-S_i,a_{S_i:t-1}\right)\\
	&=\Pr\left(X_k=s'\mid X_{S_i},t-S_i,a_{S_i:D_i-1},a_{D_i:t-1}\right)\\
	&=\Pr\left(X_k=s'\mid X_{S_i},t-D_i+Y_i,A_{i-1},a_{D_i:t-1}\right).
	\end{aligned}
	\end{equation}
	Because $t$ and $D_i$ are both known to the decision maker during the interval $t\in[D_i,D_{i+1})$, \emph{i.e.}, $t$ and $D_i$ are \textit{deterministic}, the conditional probability can be rewritten as
	\begin{equation}\label{24}
	\begin{aligned}
	\Pr\left(X_k=s'\mid X_{S_i},Y_i,A_{i-1},a_{D_i:t-1}\right).
	\end{aligned}
	\end{equation}
	Substituting (\ref{24}) into (\ref{21}) yields
	\begin{equation}\label{25}\resizebox{\hsize}{!}{$
		\begin{aligned}
		&\max_{a_{t:T}}\mathbb{E}\left[\sum_{k=t}^T \mathcal{C}(X_k, a(k))|\mathcal{I}_t\right]=\\
		&\max_{a_{t:T}}\sum_{k=t}^T \mathcal{C}(s', a(k))\cdot\Pr\left(X_k=s'\mid X_{S_i},Y_i,A_{i-1},a_{D_i:t-1}\right).
		\end{aligned}$}
	\end{equation}
	Because $a_t$ is a constant during the interval $[D_i,D_{i+1})$, we have that \begin{equation}
	a_{D_i:t-1}=a_t, \text{for } t\in [D_i,D_{i+1}).
	\end{equation}
	We can thus rewrite the RHS of (\ref{25}) as
	\begin{equation}\label{eq27}
	\begin{aligned}
	&\max_{a_{t:T}}\sum_{k=t}^T \mathcal{C}(s', a_k)\cdot\Pr\left(X_k=s'\mid X_{S_i},Y_i,A_{i-1},a_t\right)\\
	&=\max_{a_{t:T}}\mathbb{E}_{a_{t:T}}\left[\sum_{k=t}^T \mathcal{C}(X(k), a_k)\mid X_{S_i},Y_i,A_{i-1}\right].
	\end{aligned}
	\end{equation}
	From Definition \ref{definition1} and we know that $X_{S_i},Y_i,A_{i-1}$ is a sufficient statistics of $\mathcal{I}_t$ for $t\in [D_i,D_{i+1})$.
	\subsection{Determining $a_t^S$ is equivalent to Determining $Z_i$}
	By combing (\ref{eq2}) and (\ref{eq4}), we have that the sequence $\{a_t^S\}_{t\in [D_{i},D_{i+1})}$ constitutes a one-to-one mapping either to the sampling time $S_{i+1}$ or to the waiting time $Z_i$. Therefore, determining the sequence $\{a_t^S\}_{t\in [D_{i},D_{i+1})}$ is equivalent to determine $S_{i+1}$ or $Z_i$.
	
	\section{Proof of Lemma \ref{l3}}\label{appendixb}
	\subsection{Proof of Part (i)}
	\subsubsection{$h^*\le\lambda \iff U(\lambda)\le0$} 
	If $h^*\le\lambda$, there exists a policy $\pi=(Z_0,A_0,Z_1,A_1,\cdots)$ such that \begin{equation}\label{eq22}
	\lim _{\mathrm{n} \rightarrow \infty} \frac{\sum_{i=0}^{n-1} \mathbb{E}_\pi\left[\sum_{t=\it{D}_{i}}^{\it{D}_{i+1}-1} \mathcal{C}(X_t, A_i)\right]}{\sum_{i=0}^{n-1} \mathbb{E}_\pi\left[Y_{i+1}+Z_i\right]}\le\lambda,
	\end{equation}
	which is equivalent to 
	\begin{equation}\label{fracless0}
	\resizebox{1\hsize}{!}{$\begin{aligned}
		\lim _{\mathrm{n} \rightarrow \infty} \frac{\frac{1}{n}\sum_{i=0}^{n-1} \mathbb{E}_\pi\left[\sum_{t=\it{D}_{i}}^{\it{D}_{i+1}-1} \mathcal{C}(X_t, A_i)\right]-\lambda\mathbb{E}_\pi\left[Y_{i+1}+Z_i\right]}{\frac{1}{n}\sum_{i=0}^{n-1} \mathbb{E}_\pi\left[Y_{i+1}+Z_i\right]}\le0.
		\end{aligned}$}
	\end{equation}
	Since  $Y_i>0$ and $0\le Z_i<\infty$, we have that $\lim_{\mathrm{n} \rightarrow \infty}\frac{1}{n}\sum_{i=0}^{n-1} \mathbb{E}\left[Y_{i+1}+Z_i\right]$ always exists, satisfying \begin{equation}0<\label{rege}\lim_{\mathrm{n} \rightarrow \infty}\frac{1}{n}\sum_{i=0}^{n-1} \mathbb{E}\left[Y_{i+1}+Z_i\right]<\infty\end{equation} for any policies. Thus, we have that the numerator of (\ref{fracless0}) satisfies\begin{equation}\label{25eq}
	\lim _{\mathrm{n} \rightarrow \infty}\frac{1}{n}\sum_{i=0}^{n-1} \mathbb{E}_\pi\left[\sum_{t=\it{D}_{i}}^{\it{D}_{i+1}-1} \mathcal{C}(X_t, A_i)\right]-\lambda\mathbb{E}_\pi[Z_i+Y_{i+1}]\le 0.
	\end{equation}
	This implies that the infimum of the left-hand side of (\ref{25eq}) is also at most $0$, \emph{i.e}, $U(\lambda)\le 0$.
	
	On the contrary, when $U(\lambda)\le 0$, we can know that there exists a policy $\pi=(Z_0,A_0,Z_1,A_1,\cdots)$ that satisfies (\ref{25eq}). As (\ref{rege}) always holds, we can easily obtain that (\ref{eq22}) holds. Note that $h^*$ is the infimum of the left hand side of (\ref{eq22}), we have \begin{equation}
	h^*\le\lim _{\mathrm{n} \rightarrow \infty} \frac{\sum_{i=0}^{n-1} \mathbb{E}_\pi\left[\sum_{t=\it{D}_{i}}^{\it{D}_{i+1}-1} \mathcal{C}(X_t, a_i)\right]}{\sum_{i=0}^{n-1} \mathbb{E}_\pi\left[Y_{i+1}+Z_i\right]}\le\lambda.
	\end{equation}
	\subsubsection{$h^*>\lambda \iff U(\lambda)>0$}
	If $h^*>\lambda$, we have that for any policy $(Z_0,A_0,Z_1,A_1,\cdots)$, the following inequality always holds
	\begin{equation}\label{27}
	\lim _{\mathrm{n} \rightarrow \infty} \frac{\sum_{i=0}^{n-1} \mathbb{E}\left[\sum_{t=\it{D}_{i}}^{\it{D}_{i+1}-1} \mathcal{C}(X_t, A_i)\right]}{\sum_{i=0}^{n-1} \mathbb{E}\left[Y_{i+1}+Z_i\right]}>\lambda.
	\end{equation}		
	Since (\ref{rege}) holds, we have that for any policy $(Z_0,A_0,Z_1,A_1,\cdots)$, 
	\begin{equation}\label{aaa}
	\lim _{\mathrm{n} \rightarrow \infty}\frac{1}{n}\sum_{i=0}^{n-1} \mathbb{E}\left[\sum_{t=\it{D}_{i}}^{\it{D}_{i+1}-1} \mathcal{C}(X_t, A_i)\right]-\lambda\mathbb{E}[Z_i+Y_{i+1}]> 0.
	\end{equation}
	Since (\ref{aaa}) holds for any policies, it follows that the infimum value of the left-hand side (LHS) of (\ref{aaa}) is also greater than $0$, implying that $h^*>\lambda$.
	
	When $U(\lambda) > 0$, it is established that condition (\ref{aaa}) is satisfied for any policy sequence $(Z_0, A_0, Z_1, A_1, \cdots)$. Given that (\ref{rege}) always holds, it follows directly that (\ref{27}) holds for any policies, implying that the infimum of the LHS of (\ref{27}) is also greater than $\lambda$, \emph{i.e.}, $h^*>\lambda$.
	
	\subsection{Proof of Part (ii)}
	Through the proof of Part (ii), we know that $U(\lambda)=0$ is equivalent to $\lambda=h^*$. If $U(\lambda)=0$, we know that there exists an optimal policy $\pi^*=(Z_0,A_0,Z_1,A_1,\cdots)$ for Problem $\mathcal{P}4$, satisfying that 
	\begin{equation}\label{29}
	\lim _{\mathrm{n} \rightarrow \infty}\frac{1}{n}\sum_{i=0}^{n-1} \mathbb{E}\left[\sum_{t=\it{D}_{i}}^{\it{D}_{i+1}-1} \mathcal{C}(X_t, A_i)\right]-\lambda(Z_i+Y_{i+1})= 0,
	\end{equation}
	which implies that for the policy $\pi^*$, 
	\begin{equation}\label{eq23}
	\lim _{\mathrm{n} \rightarrow \infty} \frac{\sum_{i=0}^{n-1} \mathbb{E}\left[\sum_{t=\it{D}_{i}}^{\it{D}_{i+1}-1} \mathcal{C}(X_t, A_i)\right]}{\sum_{i=0}^{n-1} \mathbb{E}\left[Y_{i+1}+Z_i\right]}=\lambda.
	\end{equation}
	Note that $h^*=\lambda$, we have that for the policy $\pi^*$, \begin{equation}
	h^*=\lim _{\mathrm{n} \rightarrow \infty} \frac{\sum_{i=0}^{n-1} \mathbb{E}\left[\sum_{t=\it{D}_{i}}^{\it{D}_{i+1}-1} \mathcal{C}(X_t, A_i)\right]}{\sum_{i=0}^{n-1} \mathbb{E}\left[Y_{i+1}+Z_i\right]},
	\end{equation}
	which infers that policy $\pi^*$ is also the optimal policy of Problem $\mathcal{P}3$.
	\subsection{Proof of Part (iii)}
	From Part (i), we know that proving Part (iii) is equivalent to prove that $U(\lambda)$ is monotonically non-increasing in terms of $\lambda$, \emph{i.e.}, for any $\Delta\lambda>0$, $U(\lambda+\Delta\lambda)\le U(\lambda)$. This is verified by the following inequalities:
	\begin{equation}
	\resizebox{1\hsize}{!}{$
		\begin{aligned}	
		& U(\lambda+\Delta\lambda)=\\
		&\inf_{\phi_{0:\infty}} \lim _{\mathrm{n} \rightarrow \infty}\frac{1}{n} {\sum_{i=0}^{n-1}\left\{ \mathbb{E}\left[\sum_{t=\it{D}_{i}}^{\it{D}_{i+1}-1} \mathcal{C}(X_t, A_i)\right]-(\lambda+\Delta\lambda)\mathbb{E} \left[Z_i+Y_{i+1}\right]\right\}}\\
		&=\inf _{\phi_{0:\infty}}\left\{\lim _{\mathrm{n} \rightarrow \infty}\frac{1}{n} {\sum_{i=0}^{n-1}\left\{ \mathbb{E}\left[\sum_{t=\it{D}_{i}}^{\it{D}_{i+1}-1} \mathcal{C}(X_t, A_i)\right]-\lambda\mathbb{E} \left[Z_i+Y_{i+1}\right]\right\}}\right.\\
		&\left.-\lim _{\mathrm{n} \rightarrow \infty}\frac{1}{n}\sum_{i=0}^{n-1} \Delta\lambda\mathbb{E} [Z_i+Y_{i+1}]       \right\}\\
		&\le \inf _{\phi_{0:\infty}}\lim _{\mathrm{n} \rightarrow \infty}\frac{1}{n} {\sum_{i=0}^{n-1}\left\{ \mathbb{E}\left[\sum_{t=\it{D}_{i}}^{\it{D}_{i+1}-1} \mathcal{C}(X_t, A_i)\right]-\lambda\mathbb{E} \left[Z_i+Y_{i+1}\right]\right\}}\\
		&=U(\lambda).
		\end{aligned}$}
	\end{equation}
	Thus, we have that $\lambda=h^*$ is the unique root of $U(\lambda)=0$.
	
	\section{Transition Probability of $\mathscr{P}_{\mathrm{MDP}}(\lambda)$}\label{appendixc}
	We first calculate $\Pr\left(X_{S_{i+1}}|\mathcal{G}_i,Z_i,A_i\right)$ $\Pr\left(Y_{{i+1}}|\mathcal{G}_i,Z_i,A_i\right)$, $\Pr\left(A_{{i+1}}|\mathcal{G}_i,Z_i,A_i\right)$, respectively here. 
	\subsection{$\Pr\left(X_{S_{i+1}}|\mathcal{G}_i,Z_i,A_i\right)$}\label{CA}
	From time $S_i$ to $S_{i+1}$, the dynamics of the source $X_t$ can to divided into two Markov Sources
	\begin{itemize}
		\item During the interval $S_i \leq t \leq D_i$, the action remains constant $A_{i-1} = a$. Throughout this period, the system's dynamics adhere to a Markov chain, defined by the transition probability matrix $\mathbf{P}_{A_{i-1}}$. The transitions within this interval are quantified as $D_i - S_i = Y_i = \delta$. Thus the transition probability can be expressed by a $\delta$-step transition probability matrix, denoted as $\mathbf{P}_a^{\delta}$.
		\item In the subsequent time slot where $D_i \leq t \leq S_{i+1}$, the action is consistently $A_i$. During this phase, the system continues as another Markov chain with transition probability matrix $\mathbf{P}_{A_i}$. The number of transitions occurring in this interval is $S_{i+1} - D_i = Z_i$, which leads to a $Z_i$-step transition probability matrix $\mathbf{P}_{A_i}^{Z_i}$.
	\end{itemize}
	
	Thus, the transition probability from $t=S_i$ to $t=S_{i+1}$ is the product of the $\delta$-step transition probability matrix $\mathbf{P}_a^{\delta}$ and the $Z_i$-step transition probability matrix $\mathbf{P}_{A_i}^{Z_i}$. This yields the transition probability as:
	\begin{equation}
	\begin{aligned}
	\Pr\left(X_{S_{i+1}}=s'|\mathcal{G}_i=(s,\delta,a),Z_i,A_i\right)=[\mathbf{P}_a^{\delta}\cdot\mathbf{P}_{A_i}^{Z_i}]_{s\times s'}.
	\end{aligned}
	\end{equation}
	\subsection{$\Pr(Y_{i+1}|\mathcal{G}_i,Z_i,A_i)$}
	Note that $Y_i$ is i.i.d random variable, $Y_{i+1}$ is independent of $\mathcal{G}_i$, $Z_i$, and $A_i$. Thus, the following transition probability is almost sure
	\begin{equation}
	\begin{aligned}
	&\Pr\left(Y_{i+1}=\delta'|\mathcal{G}_i=(s,\delta,a),Z_i,A_i\right)=\Pr(Y_{i+1}=\delta').
	\end{aligned}
	\end{equation}
	\subsection{$\Pr(A_i|\mathcal{G}_i,Z_i,A_i)$}
	The state $A_i$ is actually a record of the action taken. Thus, we have the following transition probability:
	\begin{equation}
	\Pr\left(A_i=a'|\mathcal{G}_i=(s,\delta,a),Z_i,A_i\right)=\mathbbm{1}\{a'=A_i\}.
	\end{equation}
	
	At last, by leveraging the conditional independence among $X_{S_{i+1}}$, $Y_{i+1}$, and $A_i$, $\Pr(\mathcal{G}_{i+1}|\mathcal{G}_i,Z_i,A_i)$ can be decomposed into the product of
	$\Pr(X_{S_{i+1}}|\mathcal{G}_i,Z_i,A_i)$, 
	$\Pr(Y_{i+1}|\mathcal{G}_i,Z_i,A_i)$, and
	$\Pr(A_i|\mathcal{G}_i,Z_i,A_i)$, and we will have the transition probability given in (\ref{tran}).  
	
	\section{Proof of Lemma \ref{l4}}\label{appendixd}
	If the cost function is defined as in (\ref{cfun}), the objective of the MDP $\mathscr{P}_{\mathrm{MDP}}(\lambda)$ is given as
	\begin{equation}\label{36}
	\begin{aligned}
	&\inf _{\phi_{0: \infty}} \limsup _{T \rightarrow \infty} \frac{1}{T} \mathbb{E}\left[\sum_{t=1}^T g(\mathcal{G}_i, A_i,Z_i;\lambda)\right]\\
	&=\inf _{\phi_{0: \infty}} \limsup _{T \rightarrow \infty} \frac{1}{T} \mathbb{E}\left[\sum_{t=1}^T q(\mathcal{G}_i, A_i,Z_i)-\lambda\mathbb{E}[Z_i+Y_{i+1}]\right].
	\end{aligned}
	\end{equation}
	Comparing (\ref{36}) with Problem $\mathcal{P}4$, our remaining focus is to prove that $\mathbb{E}\left[\sum_{t=\it{D}_{i}}^{\it{D}_{i+1}-1} \mathcal{C}\left(X_t, A_i\right)\right]=\mathbb{E}[q(\mathcal{G}_i,A_i,Z_i)]$. 
	
	$\mathbb{E}\left[\sum_{t=\it{D}_{i}}^{\it{D}_{i+1}-1} \mathcal{C}\left(X_t, A_i\right)\right]$ can be decomposed as:
	\begin{equation}\resizebox{1\hsize}{!}{$
		\begin{aligned}
		& =\mathbb{E}\left[\underset{\substack{Y_{i+1}, \\ X_{D_i:D_{i+1}-1}}}{\mathbb{E}}\left[\left.\sum_{t=\it{D}_{i}}^{\it{D}_{i+1}-1} \mathcal{C}\left(X_t, A_i\right) \right| \mathcal{G}(i)\right]\right] \\
		&=\mathbb{E}\left[\underset{\substack{Y_{i+1}, \\ X_{D_i:D_{i+1}-1}}}{\mathbb{E}}\left[\sum_{t=\it{D}_{i}}^{\it{D}_{i}+\mathrm{Z}_{i}+\mathrm{Y}_{i+1}-1} \mathcal{C}\left(X_t, A_i\right) \mid X_{S_i}, Y_i, A_{i-1}\right]\right],
		\end{aligned}$}
	\end{equation}
	where the conditional expectation can be calculated by 
	\begin{equation}\label{38}
	\resizebox{\hsize}{!}{$
		\begin{aligned}
		&\underset{\substack{Y_{i+1}, \\ X_{D_i:D_{i+1}-1}}}{\mathbb{E}}\left[\sum_{t=\it{D}_{i}}^{\it{D}_{i}+\mathrm{Z}_{i}+\mathrm{Y}_{i+1}-1} \mathcal{C}\left(X_t, A_i\right) \mid X_{S_i}, Y_i, A_{i-1}\right]=\\
		&\underset{\substack{Y_{i+1}}}{\mathbb{E}}\left[\sum_{t=\it{D}_{i}}^{\it{D}_{i}+\mathrm{Z}_{i}+\mathrm{Y}_{i+1}-1} \sum_{s'}\mathcal{C}\left(s', A_i\right)\Pr(X_t=s'\mid X_{S_i}, Y_i, A_{i-1}) \right].
		\end{aligned}$}
	\end{equation}
	
	Up to this point, our remaining focus is to calculate the conditional probability $\Pr(X_t=s'\mid X_{S_i}, Y_i, A_{i-1})$ for $t\in[D_i,D_{i+1}-1)$. Similar our the results about $\Pr(X_{S_{i+1}}\mid\mathcal{G}_i,Z_i,A_i)$ in Appendix \ref{CA}, we also analyze the conditional probability through two steps:
	\begin{itemize}
		\item During the interval $S_i \leq t \leq D_i$, the action is $A_{i-1}$. Throughout this period, the system's source is a Markov chain, defined by the transition probability matrix $\mathbf{P}_a$. The transitions within this interval are $D_i - S_i = Y_i $. Thus the transition probability can be expressed by a $\delta$-step transition probability matrix, denoted as $\mathbf{P}_{A_i-1}^{Y_i}$.
		\item In the subsequent time slot where $D_i \leq t \leq D_{i+1}-1$, the action $A_i$. During this period, the system is another Markov chain with transition probability matrix $\mathbf{P}_{A_i}$. The number of transitions in this interval is $t-D_i$, which leads to a $(t-D_i)$-step transition probability matrix $\mathbf{P}_{A_i}^{t-D_i}$.
	\end{itemize}
	
	Thus, the transition probability from time slot $X_{S_i}$ to $X_t$, where $D_i\le t\le D_{i+1}$ is exactly the product of the $Y_i$-step transition probability matrix $\mathbf{P}_{A_{i-1}}^{Y_i}$ and the ($t-D_i$)-step transition probability matrix $\mathbf{P}_{A_i}^{t-D_i}$, given as:
	\begin{equation}\label{39}
	\Pr(X_t=s'\mid X_{S_i}, Y_i, A_{i-1})=\left[\mathbf{P}_{A_{i-1}}^{Y_i}\cdot\mathbf{P}_{A_{i}}^{t-D_i}\right]_{X_{S_i}\times s'}.
	\end{equation}
	
	By substituting (\ref{39}) into (\ref{38}), applying the variable substitution as $t:=t-D_i$, and interchanging the summation terms, we have 
	\begin{equation}\label{40}
	\resizebox{\hsize}{!}{$
		\begin{aligned}
		&\underset{\substack{Y_{i+1}}}{\mathbb{E}}\left[\sum_{t=\it{D}_{i}}^{\it{D}_{i}+\mathrm{Z}_{i}+\mathrm{Y}_{i+1}-1} \sum_{s'}\mathcal{C}\left(s', A_i\right)\cdot\left[\mathbf{P}_{A_{i-1}}^{Y_i}\cdot\mathbf{P}_{A_{i}}^{t-D_i}\right]_{X_{S_i}\times s'} \right]\\
		&=\sum_{s'}\underset{\substack{Y_{i+1}}}{\mathbb{E}}\left[\sum_{t=0}^{\mathrm{Z}_{i}+\mathrm{Y}_{i+1}-1} \cdot\left[\mathbf{P}_{A_{i-1}}^{Y_i}\cdot\mathbf{P}_{A_{i}}^{t}\right]_{X_{S_i}\times s'} \right]\cdot\mathcal{C}\left(s', A_i\right).
		\end{aligned}$}
	\end{equation}
	
	Then, note that the RHS of (\ref{40}) is exactly the definition of $q(\mathcal{G}_i,A_i,Z_i)$ as in (\ref{15}), we complete the proof that $\mathbb{E}\left[\sum_{t=\it{D}_{i}}^{\it{D}_{i+1}-1} \mathcal{C}\left(X_t, A_i\right)\right]=\mathbb{E}[q(\mathcal{G}_i,A_i,Z_i)]$. Thus the objective of the MDP $\mathscr{P}_{\mathrm{MDP}}(\lambda)$ is equivalent to Problem $\mathcal{P}4$.

	\section{Proof of Theorem \ref{the1}}\label{appendixe}
	\subsection{Existence of Optimal Stationary Policy}
	According to \cite[Proposition 6.2.3]{sennott2009stochastic}, optimal stationary policies exist under conditions where the cost function remains bounded, and both the state and action spaces are finite. 
	
	In our context, both the state space and the action space are finite space naturally. Our remaining focus is to verify that $g(\mathcal{G}_i,Z_i,A_i;\lambda)$ is bounded. Given that $h^*$ (as shown in Lemma \ref{l5}), $\lambda$, which needs to $\lambda$ to approach $h^*$, is also bounded. Because $\lambda$ is bounded, the step cost function of the MDP $\mathscr{P}_{\mathrm{MDP}}(\lambda)$, given by $g(\mathcal{G}_i,Z_i,A_i;\lambda)$ in (\ref{cfun}), is also bounded. This leads to the existence of an optimal \textit{stationary} policy. 
	
	\subsection{Existence of Optimal Deterministic Policy}
	The sufficient condition of existence of an optimal \textit{deterministic} policy is that the MDP is a \textit{unichain} \cite[Theorem 8.4.5]{puterman2014markov}. Our result is that if an MDP characterized by finite state space $\mathcal{S}$, finite action space $\mathcal{A}$, and transition probability $\mathbf{P}_a$ is a \textit{unichain}, then its corresponding remote-MDP transformation $\mathscr{P}_{\mathrm{MDP}}(\lambda)$ is also a \textit{unichain}, \emph{i.e.}, the introduction of the observation delay $\Delta(t)$ does not change the \textit{unichain} property of the MDP.
	
	The condition that an MDP characterized by finite state space $\mathcal{S}$, finite action space $\mathcal{A}$, and transition probability $\mathbf{P}_a$ is a \textit{unichain} implies that the transition probability $\mathbf{P}_a$, for any given $a\in\mathcal{A}$, is a \textit{unichain} transition probability. We call a transition probability $\mathbf{P}_a$ is a \textit{unichain} if there is a single \textit{recurrent} class plus a possibly empty set of \textit{transients} states \cite[Section 8.3.1]{puterman2014markov}\footnote{A state $s$ is said to be \textit{transient} if, starting from $s$, there is a non-zero probability that the chain will never return to $s$, \emph{i.e.}, for any positive integer $n>0$, $p_{ss}^n=0$, where $p_{ss}^n$ is the probability of that the state stars from $s$ and ends at $s$ after $n$ steps. It is called \textit{recurrent} otherwise. A \textit{recurrent class} is a group of states that \textit{communicate} with each other, \emph{i.e.}, for any $i,j\in\mathcal{X}$ where $\mathcal{X}$ is a \textit{recurrent class}, there $\lim_{n\rightarrow\infty}p_{ij}^n>0$.}. We have the following lemma:
	\begin{lemma}\label{l6}
		The following statements are true:\\
		i) If $\mathbf{P}_a$ is a \textit{unichain}, then $\mathbf{P}_a^n$, where $n\in\mathbb{N}^+$ is also a unichain;\\
		ii) If $\mathbf{P}_a$ and $\mathbf{P}_b$ are both \textit{unichains}, then $\mathbf{P}_a\cdot\mathbf{P}_b$ is a \textit{unichain}.
	\end{lemma}
	\begin{proof}
		\textit{Proof of Part i):}\\
		We first prove that if $s$ is a \textit{transient} state for Markov chain $\mathbf{P}_a$, then it is also a \textit{transient} state for $n$-step Markov chain $\mathbf{P}_a^n$. This is verified by the definition of a \textit{transient} state:
		\begin{equation}
		p_{ss}^n\equiv0,\text{ for any }n\in\mathbb{N}^+.
		\end{equation}
		
		Next, we prove that if $\mathcal{X}$ is a \textit{recurrent class} for Markov chain $\mathbf{P}_a$, then it is also a \textit{recurrent class} for $n$-step Markov chain $\mathbf{P}_a^n$. For $i,j\in\mathcal{X}$, the $m$-step transition probability from state $i$ to $j$ can be calculated by:
		\begin{equation}\label{48}
		p_{ij}^m=\sum_{k\in\mathcal{S}}\sum_{r=1}^{m-1}p_{ik}^r\cdot p_{kj}^{m-r}.
		\end{equation}
		If $k$ is a \textit{transient state}, then $p_{ik}^r\equiv0$, thus (\ref{48}) can be reduced to \begin{equation}
		\begin{aligned}
		p_{ij}^m&=\sum_{k\in\mathcal{X}}\sum_{r=1}^{m-1}p_{ik}^r\cdot p_{kj}^{m-r}\\
		&\ge\sum_{r=1}^{m-1}p_{ii}^r\cdot p_{ij}^{m-r}.
		\end{aligned}
		\end{equation}
		Since $i,j\in\mathcal{X}$, we have that there exists positive integers $N_{ii}, N_{ij}\in\mathbb{N}^+$ such that $p_{ii}^{N_{ii}}>0$ and $p_{ij}^{N_{ij}}>0$. Thus we can choose $m=N_{ii}+N_{ij}$ to establish 
		\begin{equation}\label{50}
		p_{ij}^{N_{ii}+N_{ij}}\ge\sum_{r=1}^{m-1}p_{ii}^r\cdot p_{ij}^{m-r}\ge p_{ii}^{N_{ii}}\cdot p_{ij}^{N_{ij}}>0.
		\end{equation}
		From (\ref{50}) we can deduce by \textit{mathematical induction} to ensure that \begin{equation}p_{ij}^{N_{ij}+kN_{ii}}>0, \text{ for any $k\in\mathbb{N}^+$}.\end{equation} Thus, for any $i,j\in\mathcal{X}$ and Markov chain $\mathbf{P}_a^n$ with any $n$, we can choose $nt=N_{ij}+kN_{ii}$ to hold that 
		\begin{equation}
		\lim_{t\rightarrow\infty}p_{ij}^{nt}=\lim_{k\rightarrow\infty}p_{ij}^{N_{ij}+kN_{ii}}>0,
		\end{equation}
		which implies that $i$ and $j$ are also \textit{commutative} in the Markov chain $\mathbf{P}_a^n$. \\
		\textit{Proof of Part ii):}\\	
		To prove this, we first denote the \textit{recurrent class} of $\mathbf{P}_a$ and $\mathbf{P}_b$ as $\mathcal{X}(a)$ and $\mathcal{X}(b)$, respectively. We also use $p_{ij}^n(a)$ and $p_{ij}^b$ to represent the transition probability from state $i$ to state $j$ after $n$ transition of chain $\mathbf{P}_a$ and chain $\mathbf{P}_b$, and $p_{ij}^n(ab)$ to denote the the corresponding transition probability of the chain $\mathbf{P}_a\cdot\mathbf{P}_b$. We prove Part ii) in the following cases:
		\begin{itemize}
			\item $\mathcal{X}(a)=\mathcal{X}(b)$: In such a case, for any $i,j\in\mathcal{X}(a)$, we have that \begin{equation}\label{53}
			\lim_{n\rightarrow\infty}p_{ii}^{n}(a)>0, \lim_{n\rightarrow\infty}p_{ij}^{n}(b)>0.
			\end{equation} 
			Because \begin{equation}\label{54}\begin{aligned}
			\lim_{n\rightarrow\infty}p_{ij}^{n}(ab)&=\lim_{n\rightarrow\infty}\sum_{k\in\mathcal{S}}p_{ik}^n(a)\cdot p_{kj}^n(b)\\
			&\ge\lim_{n\rightarrow\infty}\sum_{k\in\mathcal{X}(a)}p_{ik}^n(a)\cdot p_{kj}^n(b)\\
			&\ge\lim_{n\rightarrow\infty} p_{ii}^n(a)\cdot p_{ij}^n(b)>0,
			\end{aligned}
			\end{equation}
			we have that $i$ and $j$ are also \textit{communicative} for the chain $\mathbf{P}_a\cdot\mathbf{P}_b$. Thus, $\mathcal{X}(a)$ is a \textit{recurrent class} of the chain $\mathbf{P}_a\cdot\mathbf{P}_b$. We next show that if $i\in\mathcal{S}-\mathcal{X}(a)=\mathcal{S}-\mathcal{X}(b)$ is a \textit{transient} state for the chains $\mathbf{P}_a$ and chain $\mathbf{P}_b$, then $i$ is also a \textit{transient} state for the chain $\mathbf{P}_a\cdot\mathbf{P}_b$, this is verified by
			\begin{equation}\label{55}\begin{aligned}
			p_{ii}^{n}(ab)=\sum_{k\in\mathcal{S}}p_{ik}^n(a)\cdot p_{ki}^n(b).
			\end{aligned}
			\end{equation}
			Because $i$ is a \textit{transient} state of the chain $\mathbf{P}_b$, we have that $p_{ki}^n(b)=0$, and thus we can deduce from (\ref{55}) that $p_{ii}^{n}(ab)=0$. This implies that $i$ is also a \textit{transient} state of the chain $\mathbf{P}_a\cdot\mathbf{P}_b$. Thus, in case where $\mathcal{X}(a)=\mathcal{X}(b)$, the \textit{recurrent class} and the \textit{transient} state of the chain $\mathbf{P}_a\cdot\mathbf{P}_b$ remains the same as that of $\mathbf{P}_a$ and $\mathbf{P}_b$. This implies that the chain $\mathbf{P}_a\cdot\mathbf{P}_b$ is also a \textit{unichain}.
			\item $\mathcal{X}(a)\ne\mathcal{X}(b)$: In this case, we will prove that for the new chain $\mathbf{P}_a\cdot\mathbf{P}_b$, the unique \textit{recurrent class} is $\mathcal{X}(a)\cup\mathcal{X}(b)$, and the other states from $\mathcal{S}-\mathcal{X}(a)\cup\mathcal{X}(b)$ are all \textit{transient} states. There are four cases to describe $i,j\in\mathcal{X}(a)\cup\mathcal{X}(b)$:\\
			\textbf{Case 1}: $i\in\mathcal{X}(a)$ and $j\in\mathcal{X}(b)$. Similar to (\ref{54}), we rewrite $p_{ij}^{n}(ab)$ as 	\begin{equation}\label{56}\begin{aligned}
			\lim_{n\rightarrow\infty}p_{ij}^{n}(ab)&=\lim_{n\rightarrow\infty}\sum_{k\in\mathcal{S}}p_{ik}^n(a)\cdot p_{kj}^n(b)\\
			&\ge\lim_{n\rightarrow\infty}\sum_{k\in\mathcal{X}(a)\cap\mathcal{X}(b)}p_{ik}^n(a)\cdot p_{kj}^n(b).
			\end{aligned}
			\end{equation}
			Because $k\in\mathcal{X}(a)\cap\mathcal{X}(b)$, we have \begin{equation}\label{57}
			\begin{aligned}
			\lim_{n\rightarrow\infty}p_{ik}^n(a)>0, \lim_{n\rightarrow\infty}p_{kj}^n(b)>0.
			\end{aligned}
			\end{equation}
			Substituting (\ref{57}) into (\ref{56}) yields $\lim_{n\rightarrow\infty}p_{ij}^{n}(ab)>0$.
			\\
			\textbf{Case 2}. $i\in\mathcal{X}(a)$ and $j\in\mathcal{X}(a)$. From (\ref{53}) and (\ref{54}) we have that $\lim_{n\rightarrow\infty}p_{ij}^{n}(ab)>0$.\\
			\textbf{Case 3}. $i\in\mathcal{X}(b)$ and $j\in\mathcal{X}(a)$. This is a symmetry consequence of {Case 1} and thus we have that $\lim_{n\rightarrow\infty}p_{ij}^{n}(ab)>0$.\\
			\textbf{Case 4}. $i\in\mathcal{X}(b)$ and $j\in\mathcal{X}(b)$. This is a symmetry consequence of {Case 2} and thus we have that $\lim_{n\rightarrow\infty}p_{ij}^{n}(ab)>0$.
			
			Thus, we can conclude that for any $i,j\in\mathcal{X}(a)\cup\mathcal{X}(b)$, $i$ and $j$ are \textit{communicative} with each other for the chain $\mathbf{P}_a\cdot\mathbf{P}_b$, and thus $\mathcal{X}(a)\cup\mathcal{X}(b)$ is a \textit{recurrent class} of the chain $\mathbf{P}_a\cdot\mathbf{P}_b$. We next prove that states from $\mathcal{S}-\mathcal{X}(a)\cup\mathcal{X}(b)$ are all \textit{transient} states for the chain $\mathbf{P}_a\cdot\mathbf{P}_b$. Suppose $i\in\mathcal{S}-\mathcal{X}(a)\cup\mathcal{X}(b)$, we can similarly obtain (\ref{55}) and prove that $p_{ii}^n(ab)=0$. Thus, we accomplish the proof.
		\end{itemize}
	\end{proof}
	\begin{corollary}\label{c1}
		If $\mathbf{P}_a$ and $\mathbf{P}_b$ are both unichain, then for any $m,n\in\mathbb{N}^+$, $\mathbf{P}_a^m\cdot\mathbf{P}_b^n$ is also a unichain.
	\end{corollary}
	\begin{proof}
		By leveraging the results in Lemma \ref{l6}, we accomplish the proof.
	\end{proof}
	\begin{lemma}\label{l7}
		If $A_t$ and $B_t$ are both unichains and are independent with each other, then $(A_t,B_t)$ is also a unichain.
	\end{lemma}
	\begin{proof}
		We denote the \textit{recurrent class} of $A_t$ and $B_t$ as $\mathcal{A}$ and $\mathcal{B}$, respectively. We next show that the unique \textit{recurrent class} of the chain $(A_t,B_t)$ is $\mathcal{A}\times\mathcal{B}$. We denote $p_{(\alpha,\beta)(\alpha',\beta')}^n$ as the probability that the Markov chain $A_t,B_t$ moves from $\alpha',\beta'$ to $(\alpha',\beta')$ after $n$ transitions, $p_{\alpha\alpha'}^n$ and $p_{\beta\beta'}^n$ as the corresponding probability of the Markov chain $A_t$ and $B_t$, respectively. Similar to (\ref{55}), we have the following inequality:
		
		For any $(\alpha,\beta),(\alpha',\beta')\in\mathcal{A}\times\mathcal{B}$, we have
		\begin{equation}
		\lim_{n\rightarrow\infty}p_{(\alpha,\beta)(\alpha',\beta')}^n=\lim_{n\rightarrow\infty}p_{\alpha\alpha'}^n\cdot p_{\beta\beta'}^n>0.
		\end{equation}
		
		Otherwise, if $(\alpha,\beta)\notin\mathcal{A}\times\mathcal{B}$, we have
		\begin{equation}\label{58}\begin{aligned}
		p_{(\alpha,\beta)(\alpha,\beta)}^n&=\sum_{(c,d)}\sum_{r=1}^{n-1}\cdot p_{(\alpha,\beta)(c,d)}^r\cdot p_{(c,d)(\alpha,\beta)}^{n-r}\\
		&=\sum_{(c,d)}\sum_{r=1}^{n-1}p_{\alpha c}^r\cdot p_{\beta d}^r\cdot p_{c \alpha}^{n-r}\cdot p_{d\beta}^{n-r}.
		\end{aligned}
		\end{equation}
		Because $(\alpha,\beta)\notin\mathcal{A}\times\mathcal{B}$, we have that $p_{c \alpha}^{n-r}\cdot p_{d\beta}^{n-r}=0$ and thus $p_{(\alpha,\beta)(\alpha,\beta)}^n=0$. This implies that $(A_t,B_t)$ is a \textit{unichain}.
	\end{proof}
	With Corollary \ref{c1} and Lemma \ref{l7} in hand, we can now analyze the Markov decision process $\mathscr{P}_{\mathrm{MDP}}(\lambda)$. We obtain from Corollary \ref{c1} that $X_{S_i}$ is a unichain given any actions $Z_i\in\mathcal{Z}$ and $A_i\in\mathcal{A}$. In addition, $Y_i$ is also a very special type of a \textit{unichain} since all the states in $\mathcal{Y}$ are \textit{communicative} with each other. Third, $A_{i-1}$ records the actions taken in the previous step, thus, for any policies $(A_0,A_1,A_2,\cdots,)$, the available actions construct a \textit{recurrent class} and the others are \textit{transient} states. Thus the chain $A_{i-1}$ is also a \textit{unichain}. Then, from Corollary \ref{c1} we have that for any policies $(Z_0,A_0,Z_1,A_1,\cdots,)$, $\mathcal{G}_i=(X_{S_i},Y_i,A_{i-1})$ is a \textit{unichain}. Thus, MDP $\mathscr{P}_{\mathrm{MDP}}(\lambda)$ is a \textit{unichain}. This implies the existence of optimal \textit{deterministic} policies.

	\section{Proof of Lemma \ref{l5}}\label{appendixf}
	\subsection{Proof of $h^*\ge\min_{s,a}\mathcal{C}(s,a)$}
	Because $\mathcal{C}(X_t,a_t)\ge\min_{s,a}\mathcal{C}(s,a)$ always hold, we have
	\begin{equation}\label{F1}
	\begin{aligned}
	h^*&= \inf _{\phi_{0: \infty}} \limsup _{T \rightarrow \infty} \frac{1}{T} \mathbb{E}\left[\sum_{t=1}^T \mathcal{C}(X_t, a_t)\right]\\
	&\ge\inf _{\phi_{0: \infty}}\limsup _{T \rightarrow \infty} \frac{1}{T} \mathbb{E}\left[\sum_{t=1}^T \min_{s,a}\mathcal{C}(s,a)\right].
	\end{aligned}
	\end{equation}
	Because $\min_{s,a}\mathcal{C}(s,a)$ is a constant and is independent with $t$, the RHS of (\ref{F1}) is equal to $\min_{s,a}\mathcal{C}(s,a)$, and thus we have
	\begin{equation}
	h^*\ge\min_{s,a}\mathcal{C}(s,a).
	\end{equation}
	\subsection{Proof of $h^*\le\min_a\sum_{s\in\mathcal{S}}{\pi}_a(s)\cdot \mathcal{C}(s,a)$}
	\begin{figure*}[b!]
		\hrulefill
		\begin{equation}\label{70}
		\begin{aligned}
		0
		&=\min _{A_i, Z_i}\big\{g(\gamma^{\mathrm{ref}},A_i,Z_i;h^*)+\mathbb{E}[W^*\left(\gamma'\right)|\gamma^{\mathrm{ref}},Z_i,A_i]\big\}
		=\min _{A_i, Z_i}\big\{q(\gamma^{\mathrm{ref}},A_i,Z_i)-h^*\cdot f(Z_i)+\mathbb{E}[W^*\left(\gamma'\right)|\gamma^{\mathrm{ref}},Z_i,A_i]\big\}\\
		&=\min _{A_i, Z_i}\left\{f(Z_i)\cdot\left(\frac{q(\gamma^{\mathrm{ref}},A_i,Z_i)+\mathbb{E}[W^*\left(\gamma'\right)|\gamma^{\mathrm{ref}},Z_i,A_i]}{f(Z_i)}-h^*\right) \right\},
		\end{aligned}\tag{69}
		\end{equation}
	\end{figure*}
	Consider a specific policy $\phi(a)=(Z_0,a,Z_1,a,\cdots)\in\Psi$, where $\Psi\triangleq\{\phi(a):a\in\mathcal{A}\}$. Because $\Psi$ is the subset of the policy space, we have that the infimum of long-term average cost over the policy space is less than the infimum of that over $\Psi$, satisfying
	\begin{equation}\label{42}
	\begin{aligned}
	h^*&= \inf _{\phi_{0: \infty}} \limsup _{T \rightarrow \infty} \frac{1}{T} \mathbb{E}\left[\sum_{t=1}^T \mathcal{C}(X_t, a_t)\right]\\
	&\le\inf _{\phi(a)\in\Psi}\limsup _{T \rightarrow \infty} \frac{1}{T} \mathbb{E}\left[\sum_{t=1}^T \mathcal{C}(X_t,a_t)\right]\\
	&=\inf _{\phi(a)\in\Psi}\limsup _{T \rightarrow \infty} \frac{1}{T} \mathbb{E}\left[\sum_{t=1}^T \mathcal{C}(X_t,a)\right],
	\end{aligned}
	\end{equation}
	
	For a given constant action $a\in\mathcal{A}$, $X_t$ reduces to a Markov chain with transition probability $\mathbf{P}_a$, and thus the long-term average cost of the Markov chain $\limsup _{T \rightarrow \infty}\frac{1}{T} \mathbb{E}\left[\sum_{t=1}^T \mathcal{C}(X_t,a)\right]$ can be expressed by
	\begin{equation}\label{44}
	\sum_{s\in\mathcal{S}}{\pi}_a(s)\cdot \mathcal{C}(s,a),
	\end{equation}
	where $\pi_a(s)$ is the stationary of the state $s$. Denote $\boldsymbol{\pi}_a$ as the vector consisting of $\pi_a(s), s\in\mathcal{S}$, we have $\boldsymbol{\pi}_a\cdot\mathbf{P}_a=\boldsymbol{\pi}_a$. Then, substituting (\ref{44}) into (\ref{42}) yields the upper bound.
	\section{Proof of Theorem \ref{the2}}\label{appendixg}
	From \cite[Proposition 7.4.1]{bertsekas2012dynamic}, we know that for any $\lambda$, the optimal value of Problem $\mathcal{P}4$, which is $U(\lambda)$, is the same for all initial states and some values $V^*(\gamma;\lambda), \gamma\in\mathcal{S}\times\mathcal{Y}\times\mathcal{A}$ and satisfies the following Bellman equation:
	\begin{align}\label{64}
	&V^*(\gamma;\lambda)+U(\lambda)=\nonumber\\&\min _{A_i, Z_i}\big\{g(\gamma,A_i,Z_i;\lambda)+ \mathbb{E}[V^*\left(\gamma';\lambda\right)|\gamma,A_i,Z_i]\big\},
	\end{align}
	Substituting $\lambda=h^*$ and $U(h^*)=0$ into the Bellman equation, and (\ref{64}) transforms to
	\begin{align}\label{65}
	&V^*(\gamma;h^*)=\nonumber\\&\min _{A_i, Z_i}\big\{g(\gamma,A_i,Z_i;h^*)+ \mathbb{E}[V^*\left(\gamma';h^*\right)|\gamma,A_i,Z_i]\big\},
	\end{align}
	Similar to the RVI algorithm, we introduce the \textit{relative value function} defined as 
	\begin{equation}\label{66}
	W^*(\gamma)\triangleq V^*(\gamma;h^*)-V^*(\gamma^{\mathrm{ref}};h^*),
	\end{equation}
	where $\gamma^{\mathrm{ref}}$ is called \textit{reference state} and can be arbitrarily chosen from space $\mathcal{S}\times\mathcal{Y}\times\mathcal{A}$. Then, substituting (\ref{66}) into (\ref{65}) yields
	\begin{equation}\label{67}
	\begin{aligned}
	W^*(\gamma)=\min _{A_i, Z_i}\big\{g(\gamma,A_i,Z_i;h^*)+\mathbb{E}[W^*\left(\gamma'\right)|\gamma,Z_i,A_i]\big\},
	\end{aligned}
	\end{equation}
	Applying $\gamma=\gamma^{\mathrm{ref}}$ in (\ref{66}) and (\ref{67}) leads to
	\begin{equation}\label{68}
	W^*(\gamma^{\mathrm{ref}})=0,\forall \gamma\in\mathcal{S}\times\mathcal{Y}\times\mathcal{A},
	\end{equation}
	and 
	\begin{equation}\label{69}
	\begin{aligned}
	&W^*(\gamma^{\mathrm{ref}})=\\&\min _{A_i, Z_i}\big\{g(\gamma^{\mathrm{ref}},A_i,Z_i;h^*)+\mathbb{E}[W^*\left(\gamma'\right)|\gamma^{\mathrm{ref}},Z_i,A_i]\big\},
	\end{aligned}
	\end{equation}
	Then, substituting (\ref{68}) and $g(\gamma^{\mathrm{ref}},A_i,Z_i;h^*)=q(\gamma^{\mathrm{ref}},A_i,Z_i)-h^*\cdot f(Z_i)$ into (\ref{69}) yields (\ref{70}).	
	\stepcounter{equation}
	
	Because $f(Z_i)>0$, we have that (\ref{70}) holds only if
	\begin{equation}\label{71}
	\min_{A_i,Z_i}\left\{\frac{q(\gamma^{\mathrm{ref}},A_i,Z_i)+\mathbb{E}[W^*\left(\gamma'\right)|\gamma^{\mathrm{ref}},Z_i,A_i]}{f(Z_i)}-h^*\right\}=0.
	\end{equation}
	Moving $h^*$ to the RHS of (\ref{71}) yields
	\begin{equation}
	h^*=\min_{A_i,Z_i}\left\{\frac{q(\gamma^{\mathrm{ref}},A_i,Z_i)+\mathbb{E}[W^*\left(\gamma'\right)|\gamma^{\mathrm{ref}},Z_i,A_i]}{f(Z_i)}\right\}.
	\end{equation}
	We thus accomplish the proof.
	
	\section{Parameter Setting of Simulations}\label{appendixh}
	We consider the following parameter setup in the case study: 
	\begin{itemize}
		\item The \textit{state space} is a binary space: $\mathcal{S}=\{0,1\}$.
		\item The \textit{action space} is a binary space: $\mathcal{A}=\{0,1\}$.
		\item The \textit{transition probability matrix} of $X_t$ is \begin{equation}
		\mathbf{P}_0=\begin{bmatrix}
		0.9 &0.1\\0.1 &0.9
		\end{bmatrix}, \mathbf{P}_1=\begin{bmatrix}
		0.6 &0.4\\0.01 &0.99
		\end{bmatrix}.
		\end{equation}
		\item The cost function $\mathcal{C}(X_t,a_t)$ is given as \begin{equation}
		\mathcal{C}(0,0)=40,\mathcal{C}(0,1)=60,\mathcal{C}(1,0)=0,\mathcal{C}(1,1)=20.
		\end{equation}
		\item The delay is considered as a binary random delay with $\Pr(Y_i=1)=p$ and $\Pr(Y_i=10)=1-p$.
	\end{itemize}

	\bibliographystyle{IEEEtran}
	\bibliography{reference}

\begin{thebibliography}{10}
\providecommand{\url}[1]{#1}
\csname url@samestyle\endcsname
\providecommand{\newblock}{\relax}
\providecommand{\bibinfo}[2]{#2}
\providecommand{\BIBentrySTDinterwordspacing}{\spaceskip=0pt\relax}
\providecommand{\BIBentryALTinterwordstretchfactor}{4}
\providecommand{\BIBentryALTinterwordspacing}{\spaceskip=\fontdimen2\font plus
\BIBentryALTinterwordstretchfactor\fontdimen3\font minus
  \fontdimen4\font\relax}
\providecommand{\BIBforeignlanguage}[2]{{%
\expandafter\ifx\csname l@#1\endcsname\relax
\typeout{** WARNING: IEEEtran.bst: No hyphenation pattern has been}%
\typeout{** loaded for the language `#1'. Using the pattern for}%
\typeout{** the default language instead.}%
\else
\language=\csname l@#1\endcsname
\fi
#2}}
\providecommand{\BIBdecl}{\relax}
\BIBdecl

\bibitem{DBLP:journals/tit/SunUYKS17}
Y.~Sun, E.~Uysal, R.~D. Yates, C.~E. Koksal, and N.~B. Shroff, ``Update or
  wait: How to keep your data fresh,'' \emph{{IEEE} Trans. Inf. Theory},
  vol.~63, no.~11, pp. 7492--7508, 2017.

\bibitem{howard1960dynamic}
R.~A. Howard, \emph{Dynamic programming and markov processes}.\hskip 1em plus
  0.5em minus 0.4em\relax Cambridge, MA; MIT Press, 1960.

\bibitem{bellman1966dynamic}
R.~Bellman, ``Dynamic programming,'' \emph{Science}, vol. 153, no. 3731, pp.
  34--37, 1966.

\bibitem{boucherie2017markov}
R.~J. Boucherie and N.~M. Van~Dijk, \emph{Markov decision processes in
  practice}.\hskip 1em plus 0.5em minus 0.4em\relax Cham, Switzerland:
  Springer, 2017.

\bibitem{DBLP:conf/sigmetrics/AltmanN92}
E.~Altman and P.~Nain, ``Closed-loop control with delayed information,''
  \emph{Perf. Eval. Rev.}, vol.~14, pp. 193--204, 1992.

\bibitem{DBLP:journals/tac/KatsikopoulosE03}
K.~V. Katsikopoulos and S.~E. Engelbrecht, ``Markov decision processes with
  delays and asynchronous cost collection,'' \emph{{IEEE} Trans. Autom.
  Control.}, vol.~48, no.~4, pp. 568--574, 2003.

\bibitem{DBLP:conf/isit/Yates15}
R.~D. Yates, ``Lazy is timely: Status updates by an energy harvesting source,''
  in \emph{{IEEE} Proc. {ISIT}}, 2015, pp. 3008--3012.

\bibitem{DBLP:journals/tcom/ArafaBSP21}
A.~Arafa, K.~Banawan, K.~G. Seddik, and H.~V. Poor, ``Sample, quantize, and
  encode: Timely estimation over noisy channels,'' \emph{{IEEE} Trans.
  Commun.}, vol.~69, no.~10, pp. 6485--6499, 2021.

\bibitem{DBLP:journals/tit/TangCWYT23}
H.~Tang, Y.~Chen, J.~Wang, P.~Yang, and L.~Tassiulas, ``Age optimal sampling
  under unknown delay statistics,'' \emph{{IEEE} Trans. Inf. Theory}, vol.~69,
  no.~2, pp. 1295--1314, 2023.

\bibitem{panjiayu2023}
J.~Pan, A.~M. Bedewy, Y.~Sun, and N.~B. Shroff, ``{Optimal sampling for data
  freshness: Unreliable transmissions with random two-way delay},''
  \emph{IEEE/ACM Trans. Netw.}, vol.~31, no.~1, pp. 408--420, 2023.

\bibitem{BZJSSU}
B.~Zhou and W.~Saad, ``{Joint status sampling and updating for minimizing Age
  of Information in the Internet of Things},'' \emph{IEEE Trans. Commun.},
  vol.~67, no.~11, pp. 7468--7482, 2019.

\bibitem{DBLP:conf/infocom/KaulYG12}
S.~K. Kaul, R.~D. Yates, and M.~Gruteser, ``Real-time status: How often should
  one update?'' in \emph{Proc. IEEE INFOCOM 2012}, 2012, pp. 2731--2735.

\bibitem{DBLP:journals/ftnet/KostaPA17}
A.~Kosta, N.~Pappas, and V.~Angelakis, ``{Age of Information: {A} new concept,
  metric, and tool},'' \emph{Found. Trends Netw.}, vol.~12, no.~3, pp.
  162--259, 2017.

\bibitem{DBLP:journals/jsac/YatesSBKMU21a}
R.~D. Yates, Y.~Sun, D.~R. Brown, S.~K. Kaul, E.~H. Modiano, and S.~Ulukus,
  ``{Age of Information: An Introduction and Survey},'' \emph{{IEEE} J. Sel.
  Areas Commun.}, vol.~39, no.~5, pp. 1183--1210, 2021.

\bibitem{10105150}
J.~Cao, X.~Zhu, S.~Sun, Z.~Wei, Y.~Jiang, J.~Wang, and V.~K. Lau, ``{Toward
  industrial metaverse: Age of Information, latency and reliability of
  short-packet transmission in 6G},'' \emph{IEEE Wirel. Commun.}, vol.~30,
  no.~2, pp. 40--47, 2023.

\bibitem{costa2016age}
M.~Costa, M.~Codreanu, and A.~Ephremides, ``On the age of information in status
  update systems with packet management,'' \emph{IEEE Trans. Inf. Theory},
  vol.~62, no.~4, pp. 1897--1910, 2016.

\bibitem{DBLP:journals/tcom/DoganA21}
O.~Dogan and N.~Akar, ``{The multi-source probabilistically preemptive
  {M/PH/1/1} Queue With packet errors},'' \emph{{IEEE} Trans. Commun.},
  vol.~69, no.~11, pp. 7297--7308, 2021.

\bibitem{yates2018age}
R.~D. Yates and S.~K. Kaul, ``{The Age of Information: Real-time status
  updating by multiple sources},'' \emph{IEEE Trans. Inf. Theory}, vol.~65,
  no.~3.

\bibitem{kam2015effect}
C.~Kam, S.~Kompella, G.~D. Nguyen, and A.~Ephremides, ``Effect of message
  transmission path diversity on status age,'' \emph{IEEE Trans. Inf. Theory},
  vol.~62, no.~3, pp. 1360--1374, 2015.

\bibitem{sun2019samplingwiener}
Y.~Sun, Y.~Polyanskiy, and E.~Uysal, ``Sampling of the {Wiener} process for
  remote estimation over a channel with random delay,'' \emph{IEEE Trans. Inf.
  Theory}, vol.~66, no.~2, pp. 1118--1135, 2019.

\bibitem{ornee2019sampling}
T.~Z. Ornee and Y.~Sun, ``{Sampling for remote estimation through queues: Age
  of Information and beyond},'' in \emph{IEEE Proc. WiOPT}, 2019, pp. 1--8.

\bibitem{10.1145/3492866.3549732}
H.~Tang, Y.~Sun, and L.~Tassiulas, ``Sampling of the wiener process for remote
  estimation over a channel with unknown delay statistics,'' in \emph{Proc. ACM
  MobiHoc}, 2022, p. 51–60.

\bibitem{tsai2021unifying}
C.-H. Tsai and C.-C. Wang, ``{Unifying AoI minimization and remote
  estimation—Optimal sensor/controller coordination with random two-way
  delay},'' \emph{IEEE/ACM Trans. Netw.}, vol.~30, no.~1, pp. 229--242, 2021.

\bibitem{ornee2021sampling}
T.~Z. Ornee and Y.~Sun, ``{Sampling and remote estimation for the
  Ornstein-Uhlenbeck process through queues: Age of Information and beyond},''
  \emph{IEEE/ACM Trans. Netw.}, vol.~29, no.~5, pp. 1962--1975, 2021.

\bibitem{li2022age}
A.~Li, S.~Wu, J.~Jiao, N.~Zhang, and Q.~Zhang, ``{Age of Information with
  Hybrid-ARQ: A Unified Explicit Result},'' \emph{IEEE Trans. Commun.},
  vol.~70, no.~12, pp. 7899--7914, 2022.

\bibitem{pan2022age}
H.~Pan, T.-T. Chan, V.~C. Leung, and J.~Li, ``{Age of Information in
  Physical-layer Network Coding Enabled Two-way Relay Networks},'' \emph{IEEE
  Trans. Mob. Comput.}, 2022.

\bibitem{xie2020age}
M.~Xie, Q.~Wang, J.~Gong, and X.~Ma, ``{Age and energy analysis for LDPC coded
  status update with and without ARQ},'' \emph{IEEE Internet Things J.},
  vol.~7, no.~10, pp. 10\,388--10\,400, 2020.

\bibitem{meng2022analysis}
S.~Meng, S.~Wu, A.~Li, J.~Jiao, N.~Zhang, and Q.~Zhang, ``{Analysis and
  optimization of the HARQ-based Spinal coded timely status update system},''
  \emph{IEEE Trans. Commun.}, vol.~70, no.~10, pp. 6425--6440, 2022.

\bibitem{DBLP:journals/tcom/CaoZJWS21}
J.~Cao, X.~Zhu, Y.~Jiang, Z.~Wei, and S.~Sun, ``{Information age-delay
  correlation and optimization with Finite Block Length},'' \emph{{IEEE} Trans.
  Commun.}, vol.~69, no.~11, pp. 7236--7250, 2021.

\bibitem{DBLP:conf/globecom/Long0GLN22}
Y.~Long, W.~Zhang, S.~Gong, X.~Luo, and D.~Niyato, ``{AoI-aware scheduling and
  trajectory optimization for multi-UAV-assisted wireless networks},'' in
  \emph{{IEEE} {GLOBECOM}}, 2022, pp. 2163--2168.

\bibitem{DBLP:journals/tcom/FengWFCD24}
H.~Feng, J.~Wang, Z.~Fang, J.~Chen, and D.~Do, ``Evaluating {AoI}-centric
  {HARQ} protocols for {UAV} networks,'' \emph{{IEEE} Trans. Commun.}, vol.~72,
  no.~1, pp. 288--301, 2024.

\bibitem{DBLP:journals/tmc/PanCLL23}
H.~Pan, T.~Chan, V.~C.~M. Leung, and J.~Li, ``{Age of Information in
  physical-layer network coding enabled two-way relay networks},'' \emph{{IEEE}
  Trans. Mob. Comput.}, vol.~22, no.~8, pp. 4485--4499, 2023.

\bibitem{10539623}
Y.~Long, S.~Zhao, S.~Gong, B.~Gu, D.~Niyato, and X.~Shen, ``{AoI}-aware sensing
  scheduling and trajectory optimization for multi-{UAV}-assisted wireless
  backscatter networks,'' \emph{IEEE Trans. Veh. Technol.}, pp. 1--16, 2024.

\bibitem{10273599}
D.~Gündüz, F.~Chiariotti, K.~Huang, A.~E. Kalør, S.~Kobus, and P.~Popovski,
  ``Timely and massive communication in {6G}: Pragmatics, learning, and
  inference,'' \emph{IEEE BITS Inf. Theory Mag.}, vol.~3, no.~1, pp. 27--40,
  2023.

\bibitem{haas2006stochastic}
P.~J. Haas, \emph{Stochastic petri nets: Modelling, stability,
  simulation}.\hskip 1em plus 0.5em minus 0.4em\relax Springer Science \&
  Business Media, 2006.

\bibitem{dinkelbach1967nonlinear}
W.~Dinkelbach, ``On nonlinear fractional programming,'' \emph{Management
  science}, vol.~13, no.~7, pp. 492--498, 1967.

\bibitem{bertsekas2012dynamic2}
D.~Bertsekas, \emph{{Dynamic Programming and Optimal Control: Volume
  II}}.\hskip 1em plus 0.5em minus 0.4em\relax Athena scientific, 2012, vol.~2.

\bibitem{hildebrand1987introduction}
F.~B. Hildebrand, \emph{Introduction to numerical analysis}.\hskip 1em plus
  0.5em minus 0.4em\relax Courier Corporation, 1987.

\bibitem{wang2023review}
Y.~Wang, S.~Wu, C.~Lei, J.~Jiao, and Q.~Zhang, ``A review on wireless networked
  control system: The communication perspective,'' \emph{{IEEE} Internet Things
  J.}, 2023.

\bibitem{9551200}
N.~Pappas and M.~Kountouris, ``Goal-oriented communication for real-time
  tracking in autonomous systems,'' in \emph{IEEE Proc. ICAS}, 2021, pp. 1--5.

\bibitem{10560514}
S.~Meng, S.~Wu, J.~Zhang, J.~Cheng, H.~Zhou, and Q.~Zhang,
  ``Semantics-empowered space-air-ground-sea integrated network: New paradigm,
  frameworks, and challenges,'' \emph{IEEE Commun. Surv. Tutorials}, pp. 1--1,
  2024.

\bibitem{9475174}
M.~Kountouris and N.~Pappas, ``Semantics-empowered communication for networked
  intelligent systems,'' \emph{IEEE Commun. Mag.}, vol.~59, no.~6, pp. 96--102,
  2021.

\bibitem{9919752}
E.~Uysal, O.~Kaya, A.~Ephremides, J.~Gross, M.~Codreanu, P.~Popovski,
  M.~Assaad, G.~Liva, A.~Munari, B.~Soret, T.~Soleymani, and K.~H. Johansson,
  ``Semantic communications in networked systems: A data significance
  perspective,'' \emph{IEEE Netw.}, vol.~36, no.~4, pp. 233--240, 2022.

\bibitem{DBLP:journals/corr/abs-2311-11143}
\BIBentryALTinterwordspacing
C.~Ari, M.~K.~C. Shisher, E.~Uysal, and Y.~Sun, ``Goal-oriented communications
  for remote inference with two-way delay,'' \emph{arXiv}, 2023. [Online].
  Available: \url{https://doi.org/10.48550/arXiv.2311.11143}
\BIBentrySTDinterwordspacing

\bibitem{10579545}
A.~Li, S.~Wu, S.~Meng, R.~Lu, S.~Sun, and Q.~Zhang, ``Toward goal-oriented
  semantic communications: New metrics, framework, and open challenges,''
  \emph{IEEE Wirel. Commun.}, pp. 1--8, 2024.

\bibitem{10562359}
A.~Li, S.~Wu, S.~Sun, and J.~Cao, ``Goal-oriented tensor: Beyond age of
  information towards semantics-empowered goal-oriented communications,''
  \emph{IEEE Trans. Commun.}, pp. 1--1, 2024.

\bibitem{10409276}
M.~Salimnejad, M.~Kountouris, and N.~Pappas, ``Real-time reconstruction of
  markov sources and remote actuation over wireless channels,'' \emph{IEEE
  Trans. Commun.}, pp. 1--1, 2024.

\bibitem{sun2019sampling}
Y.~Sun and B.~Cyr, ``{Sampling for data freshness optimization: Non-linear age
  functions},'' \emph{J. Commun. Netw.}, vol.~21, no.~3, pp. 204--219, 2019.

\bibitem{sennott2009stochastic}
L.~I. Sennott, \emph{Stochastic dynamic programming and the control of queueing
  systems}.\hskip 1em plus 0.5em minus 0.4em\relax John Wiley \& Sons, 2009.

\bibitem{puterman2014markov}
M.~L. Puterman, \emph{Markov decision processes: discrete stochastic dynamic
  programming}.\hskip 1em plus 0.5em minus 0.4em\relax John Wiley \& Sons,
  2014.

\bibitem{bertsekas2012dynamic}
D.~Bertsekas, \emph{{Dynamic Programming and Optimal Control: Volume I}}.\hskip
  1em plus 0.5em minus 0.4em\relax Athena scientific, 2012, vol.~1.

\end{thebibliography}
	
\end{document}